\documentclass[11pt]{article}
\usepackage{amssymb,amsmath,amsthm}
\usepackage{amsfonts}
\usepackage{graphicx}
\usepackage{pgf}
\usepackage{tikz}
\usepackage{pstricks}
\usepackage[font=footnotesize]{caption}
\usepackage{color}
\usepackage{multirow}
\usepackage{lscape}
\usepackage{rotating}
\usepackage[font=footnotesize]{caption}
\usepackage{graphicx}
\usepackage{natbib}
\usepackage{appendix}
\usepackage{setspace}
\usepackage{booktabs}
\usepackage[top=1in,bottom=1in,left=1in,right=1in]{geometry}

\DeclareGraphicsExtensions{.pdf,.png,.jpg}
\usetikzlibrary{shapes,snakes}
\DeclareGraphicsExtensions{.pdf,.png,.jpg}
\newtheorem{theorem}{Theorem}

\newtheorem{lemma}{Lemma}

\numberwithin{equation}{section}
\input{tcilatex}
\begin{document}

\title{Nonparametric prediction with spatial data}
\author{Abhimanyu Gupta \thanks{%
Department of Economics, University of Essex, Wivenhoe Park, Colchester CO4
3SQ, U.K. Email: a.gupta@essex.co.uk} \thanks{%
Research supported by ESRC grant ES/R006032/1.} \and Javier Hidalgo \thanks{%
Economics Department, London School of Economics, Houghton Street, London
WC2A 2AE, U.K. Email: f.j.hidalgo@lse.ac.uk}}
\date{\today }
\maketitle

\begin{abstract}
We describe a (nonparametric) prediction algorithm for spatial data, based
on a canonical factorization of the spectral density function. We provide
theoretical results showing that the predictor has desirable asymptotic
properties. Finite sample performance is assessed in a Monte Carlo study
that also compares our algorithm to a rival nonparametric method based on
the infinite $AR$ representation of the dynamics of the data. Finally, we
apply our methodology to predict house prices in Los Angeles. \newline
\end{abstract}

\section{\textbf{Introduction}}

\label{sec:intro}

Random models for spatial or spatio-temporal data play an important role in
many disciplines of economics, such as environmental, urban, development or
agricultural economics as well as economic geography, among others. When
data is collected over time such models are termed `noncausal' and have
drawn interest in economics, see for instance \cite{Breidt2001} among others
for some early examples. Other studies may be found in the special volume by 
\cite{Baltagi2007} or \cite{Cressie}. Classic treatments include the work by 
\cite{Mercer1911} on wheat crop yield data (see also \cite{Gao2006}) or \cite%
{Batchelor1918} which was employed as an example and analysed in the
celebrated paper by \cite{Whittle1954}. Other illustrations are given in 
\cite{Cressie1999}, see also \cite{Fernandez-Casal2003}. With a view towards
applications in environmental and agricultural economics, \cite{Mitchell2005}
employed a model of the type studied in this paper to analyse the effect of
carbon dioxide on crops, whereas \cite{Genton2008} examine the yield of
barley in UK. The latter manuscripts shed light on how these models can be
useful when there is evidence of spatial movement, such as that of
pollutants, due to winds or ocean currents.

Doubtless one of the main aims when analysing data is to provide predicted
values of realizations of the process. More specifically, assume that we
have a realization $\mathcal{X}_{n}=\left\{ x_{t_{i}}\right\} _{i=1}^{n}$ at
locations $t_{1},...,t_{n}$ of a process $\left\{ x_{t}\right\} _{t\in 
\mathcal{D}}$, where $\mathcal{D}$ is a subset of $\mathbb{R}^{d}$. We wish
then to predict the value of $x_{t}$ at some unobserved location $t_{0}$,
say $x_{t_{0}}$. For instance in a time series context, we wish to predict
the value $x_{n+1}$ at the unobserved location (future time) $n+1$ given a
stretch of data $x_{1},..,x_{n}$. It is often the case that the predictor of 
$x_{t_{0}}$ is based on a weighted average of the data $\mathcal{X}_{n}$,
that is%
\begin{equation}
\widehat{x}_{t_{0}}=\sum_{i=1}^{n}\beta _{i}x_{t_{i}}\text{,}  \label{pred1}
\end{equation}%
where the weights $\beta _{1},...,\beta _{n}$ are chosen to minimize the $%
\mathcal{L}_{2}$-risk function 
\begin{equation*}
E\left( x_{t_{0}}-\sum_{i=1}^{n}b_{i}x_{t_{i}}\right) ^{2}
\end{equation*}%
with respect to $b_{1},...,b_{n}$. With spatial data, the solution in $%
\left( \ref{pred1}\right) $ is referred as the Kriging predictor, see \cite%
{Stein1999}, which is also the best linear predictor for $x_{t_{0}}$. Notice
that under Gaussianity or our Condition $C1$ below, the best linear
predictor is also the best predictor. It is important to bear in mind that
with spatial data prediction is also associated with both interpolation as
well as extrapolation.

The optimal weights $\left\{ \beta _{i}\right\} _{i=1}^{n}$ in $\left( \ref%
{pred1}\right) $ depend on the covariogram (or variogram) structure of $%
\left\{ x_{t_{1}},...,x_{t_{n}};x_{t_{0}}\right\} =:\left\{ \mathcal{X}%
_{n};x_{t_{0}}\right\} $, see among others \cite{Stein1999} or \cite{Cressie}%
. That is, denoting the covariogram by $Cov\left( x_{t_{i}},x_{t_{j}}\right)
=:C\left( t_{i},t_{j}\right) $ and assuming stationarity, so that $C\left(
t_{i},t_{j}\right) =:C\left( \left\vert t_{i}-t_{j}\right\vert \right) $, we
have that the best linear predictor $\left( \ref{pred1}\right) $ becomes%
\begin{equation}
\widehat{x}_{t_{0}}=\gamma ^{\prime }\left( t_{0}\right) \boldsymbol{C}%
\mathcal{X}_{n}\text{,}  \label{pred2}
\end{equation}%
where%
\begin{equation*}
\boldsymbol{C}=\left\{ C\left( \left\vert t_{i}-t_{j}\right\vert \right)
\right\} _{i,j=1}^{n};\text{ \ \ \ }\gamma ^{\prime }\left( t_{0}\right)
=Cov\left( \mathcal{X}_{n};x_{t_{0}}\right) =E\left\{ x_{t_{0}}\left(
x_{t_{1}},...,x_{t_{n}}\right) \right\} \text{.}
\end{equation*}

When the data is regularly observed, the unknown covariogram function $%
C\left( h\right) $ is replaced by its sample analogue%
\begin{equation*}
\widehat{C}\left( h\right) =\frac{1}{\left\vert n\left( h\right) \right\vert 
}\sum_{n\left( h\right) }x_{t_{i}}x_{t_{j}}\text{,}
\end{equation*}%
where $n\left( h\right) =\left\{ \left( t_{i},t_{j}\right) :\left\vert
t_{i}-t_{j}\right\vert =h\right\} $ and $\left\vert n\left( h\right)
\right\vert $ denotes the cardinality of the set $n\left( h\right) $. When
the data is not regularly spaced some modifications of $\widehat{C}\left(
h\right) $ have been suggested, see Cressie $\left( 1993,p.70\right) $ for
details. One problem with the above estimator $\widehat{C}\left( h\right) $
is that it can only be employed for lags $h$ which are found in the data,
and hence the Kriging predictor $\left( \ref{pred2}\right) $ cannot be
computed if $\left\vert t_{i}-t_{0}\right\vert \not=h$ for any $h$ such that 
$n\left( h\right) $ is not an empty set. To avoid this problem a typical
solution is to assume some specific parametric function $C\left( h\right)
=:C\left( h;\theta \right) $, so that one computes $\left( \ref{pred2}%
\right) $ with $C\left( h;\widehat{\theta }\right) $ replacing $C\left(
h\right) $ therein, where $\widehat{\theta }$ is some estimator of $\theta $.

In this paper, we shall consider the situation when the spatial data is
collected regularly, that is on a lattice. This may occur as a consequence
of some planned experiment or due to a systematic sampling scheme, or when
we can regard the (possibly non-gridded) observations as the result of
aggregation over a set of covering regions rather than values at a
particular site, see e.g. \cite{Conley1999}, \cite{Conley2007a}, \cite%
{Bester2011}, \cite{Wang2013}, \cite{Nychka2015}, \cite{Bester2014}. As a
result of this ability to map locations to a regular grid, lattice data are
frequently studied in the econometrics literature, see e.g. \cite%
{Roknossadati2010}, \cite{Robinson2011} and \cite{Jenish2016}. Nonsystematic
patterns may occur, although these might arise as a consequence of missing
observations, see \cite{Jenish2012} for a study that covers irregular
spatial data.

However contrary to the solution given in $\left( \ref{pred2}\right) $, our
aim is to provide an estimator of $\left( \ref{pred1}\right) $ without
assuming any particular parameterization of the dynamic or covariogram
structure of the data a priori, for instance without assuming any particular
functional form for the covariogram $C\left( h\right) $. The latter might be
of interest as we avoid the risk that misspecification might induce on the
predictor. In this sense, this paper may be seen as a spatial analog of
contributions in a standard time series context such as \cite{Bhansali1974}
and \cite{Hidalgo2002}.

The remainder of the paper is organized as follows. In the next section, we
describe the multilateral and unilateral representation of the data and
their links with a Wold-type decomposition. We also describe the canonical
factorization of the spectral density function, which plays an important
role in our prediction methodology described in Section \ref{sec:FE_pred},
wherein we examine its statistical properties. Section \ref{sec:MC}
describes a small Monte-Carlo experiment to gain some information regarding
the finite sample properties of the algorithm, and compares our frequency
domain predictor to a potential `space-domain' competitor. Because land
value and real-estate prices comprise classical applications of spatial
methods, see e.g. \cite{IversenJr2001}, \cite{Banerjee2004}, \cite%
{Majumdar2006}, we apply the procedures to prediction of house prices in Los
Angeles in Section \ref{sec:empirical}. Finally, Section \ref{sec:conc}
gives a summary of the paper whereas the proofs are confined to the
mathematical appendix.

\section{Multilateral and unilateral representations}

Before we describe how to predict the value of the process $\left\{
x_{t}\right\} _{t\in \mathbb{Z}^{d}}$ at unobserved locations, for $d\geq 1$%
, it is worth discussing what do we understand by multilateral and
unilateral representations of the process and, more importantly, the link
with the Wold-type decomposition. Recall that in the prediction theory of
stationary time series, i.e. when $d=1$, the Wold decomposition plays a key
role. For that purpose, and using the notation that for any $a\in \mathbb{Z}%
^{d}$, $a=\left( a\left[ 1\right] ,...,a\left[ d\right] \right) $, so that $%
t-j$ stands for $\left( t\left[ 1\right] -j\left[ 1\right] ,....,t\left[ d%
\right] -j\left[ d\right] \right) $, we shall assume that the (spatial)
process $\left\{ x_{t}\right\} _{t\in \mathbb{Z}^{d}}$ admits a
representation given by%
\begin{equation}
x_{t}=\sum_{j\in \mathbb{Z}^{d}}\psi _{j}\varepsilon _{t-j},\text{ }\ \text{ 
}\sum_{j\in \mathbb{Z}^{d}}\left\{ \sum_{\ell =1}^{d}j^{2}\left[ \ell \right]
\right\} \left\vert \psi _{j}\right\vert <\infty \text{,}  \label{a1}
\end{equation}%
where the $\varepsilon _{t}$ are independent and identically distributed
random variables with zero mean, unit variance and finite fourth moments.
The model in $\left( \ref{a1}\right) $ denotes the dynamics of $x_{t}$ and
it is known as the multilateral representation of $\left\{ x_{t}\right\}
_{t\in \mathbb{Z}^{d}}$. It is worth pointing that a consequence of the
latter representation is that the sequence $\left\{ \varepsilon _{t}\right\}
_{t\in \mathbb{Z}^{d}}$ loses its interpretation as being the
\textquotedblleft prediction\textquotedblright\ error of the model, and thus
they can no longer be regarded as innovations, as was first noticed by \cite%
{Whittle1954}. When $d=1$, this multilateral representation gives rise to
so-called noncausal models or, in \citeauthor{Whittle1954}'s terminology,
linear \emph{transect} models. These models can be regarded as forward
looking and have gained some consideration in economics, see for instance 
\cite{Lanne2011}, \cite{Davis2013}, \cite{Lanne2013} or \cite{Cavaliere2018}.

It is worth remarking that, contrary to $d=1$, it is not sufficient for the
coefficients $\psi _{j}$ in $\left( \ref{a1}\right) $ to be $O\left(
\left\vert j\right\vert ^{-3-\eta }\right) $ for any $\eta >0$ as our next
example illustrates. Indeed, suppose that $\psi _{j}=\left( j\left[ 1\right]
+j\left[ 2\right] \right) ^{-4}=O\left( \left\Vert j\right\Vert ^{-4}\right) 
$. However it is known that the sequence $\left\{ \sum_{\ell =1}^{d}j^{2}%
\left[ \ell \right] \right\} \left\vert \psi _{j}\right\vert $ is not
summable. That is, see for instance \cite{Limaye2009}, 
\begin{equation}
c_{N}^{-1}=\left( \sum_{j\left[ 1\right] ,j\left[ 2\right] =1}^{N}\left\{
\sum_{\ell =1}^{d}j^{2}\left[ \ell \right] \right\} \left\vert \psi
_{j}\right\vert \right) ^{-1}\underset{N\rightarrow \infty }{\rightarrow }0%
\text{.}  \label{decay}
\end{equation}

One classical parameterization of $\left( \ref{a1}\right) $ is the $ARMA$
field model 
\begin{eqnarray*}
P\left( L\right) x_{t} &=&Q\left( L\right) \varepsilon _{t}\text{,} \\
P\left( z\right) &=&\sum_{j\in \mathbb{Z}_{1}^{d}}\alpha _{j}z^{j};\quad
\alpha _{0}=1;\text{ \ }Q\left( z\right) =\sum_{j\in \mathbb{Z}%
_{2}^{d}}\beta _{j}z^{j};\quad \beta _{0}=1\text{,}
\end{eqnarray*}%
where $\mathbb{Z}_{1}^{d}$ and $\mathbb{Z}_{2}^{d}$ are finite subsets of $%
\mathbb{Z}^{d}$ and henceforth $z^{j}=\prod\nolimits_{\ell =1}^{d}z\left[
\ell \right] ^{j\left[ \ell \right] }$ with the convention that $0^{0}=1$.
As an example, we have the $ARMA\left( -k_{1},k_{2};-\ell _{1},\ell
_{2}\right) $ field 
\begin{equation}
\sum_{j=-k_{1}}^{k_{2}}\alpha _{j}x_{t-j}=\sum_{j=-\ell _{1}}^{\ell
_{2}}\beta _{j}\varepsilon _{t-j},\text{ \ \ \ }\alpha _{0}=\beta _{0}=1%
\text{.}  \label{arma}
\end{equation}

As mentioned above, the Wold decomposition, and hence the concept of past
and future, plays a key role in the theory of prediction when $d=1$.
However, contrary to the situation when $d=1$, an intrinsic problem with
spatial or lattice data is that we cannot assign a unique meaning to the
concept of \textquotedblleft past\textquotedblright\ and/or
\textquotedblleft future\textquotedblright . One immediate consequence is
then that different definitions of what might be considered as past (or
future) will yield different Wold-type decompositions. More specifically,
denote a \textquotedblleft half-plane\textquotedblright\ of $\mathbb{Z}^{2}$
according to the lexicographical (dictionary) ordering \textquotedblleft $%
\prec $\textquotedblright\ defined as 
\begin{equation}
j\prec k\Leftrightarrow \left( j\left[ 1\right] <k\left[ 1\right] \right) 
\text{ or }\left( j\left[ 1\right] =k\left[ 1\right] \vee j\left[ 2\right] <k%
\left[ 2\right] \right) \text{,}  \label{lex_1}
\end{equation}%
where herewith we shall consider the case when $d=2$, often encountered with
real data. The half-plane defined by \textquotedblleft $\prec $%
\textquotedblright\ is illustrated in Figure \ref{fig:hp_ill}. Following
earlier work by \cite{Helson1958,Helson1961}, there exists then a Wold-type
representation of the (spatial) process $\left\{ x_{t}\right\} _{t\in 
\mathbb{Z}^{2}}$ given by 
\begin{equation}
x_{t}=\vartheta _{t}+\sum_{0\prec j}\zeta _{j}\vartheta _{t-j},\text{ }\ 
\text{ }\sum_{0\prec j}\left\vert \zeta _{j}\right\vert <\infty \text{,}
\label{uni_1}
\end{equation}%
where $\left\{ \vartheta _{t}\right\} _{t\in \mathbb{Z}^{2}}$ is a zero mean
white noise sequence with finite second moments $\sigma _{\vartheta }^{2}$.
It is worth recalling that $\vartheta _{t}$ once again has the
interpretation of being the \textquotedblleft one-step\textquotedblright\
prediction error. Often $\left( \ref{uni_1}\right) $ is called a unilateral
representation of $\left\{ x_{t}\right\} _{t\in \mathbb{Z}^{2}}$ as opposed
to the multilateral representation in $\left( \ref{a1}\right) $. See also 
\cite{Whittle1954} for some earlier work on multilateral versus unilateral
representations. As an example, $\left( \ref{arma}\right) $ becomes a
unilateral or causal model when $\ell _{1}=k_{1}=0$. $\left( \ref{uni_1}%
\right) $ might be regarded as a particular way to model the dependence of $%
x_{t}$ induced by the lexicographic ordering in $\left( \ref{lex_1}\right) $%
. Of course, the choice of the \textquotedblleft
half-plane\textquotedblright\ of $\mathbb{Z}^{2}$ according to the
associated chosen lexicographic ordering is not the only possible one. That
is, a different choice of \textquotedblleft half-plane\textquotedblright\ of 
$\mathbb{Z}^{2}$, induced by the lexicographic ordering, will yield a
\textquotedblleft similar\textquotedblright\ but different representation of 
$x_{t}$ to that given in $\left( \ref{uni_1}\right) $. As it will become
clear in the next section, the choice of a specific lexicographic ordering,
or its associated half-plane, will depend very much on practical purposes.
For instance, the choice of $\left( \ref{lex_1}\right) $ will depend on the
location where we wish to predict $x_{t}$. Last but not least it is worth,
and important, mentioning that the sequences $\left\{ \varepsilon
_{t}\right\} _{t\in \mathbb{Z}^{2}}$ and $\left\{ \vartheta _{t}\right\}
_{t\in \mathbb{Z}^{2}}$ are not the same. Recall that a similar phenomenon
occurs when $d=1$ and the practitioner allows for noncausal/bilateral
representations of the sequence $x_{t}$. When this is the case, the
\textquotedblleft bilateral or noncausal\textquotedblright\ representation
has errors which are independent and identically distributed, whereas for
its \textquotedblleft unilateral or causal\textquotedblright representation,
the corresponding errors are only a white noise sequence.

It is clear from the introduction that to provide accurate and valid
(linear) predictions (or interpolations), a key component is to obtain the
covariogram function of the sequence $\left\{ x_{t}\right\} _{t\in \mathbb{Z}%
^{2}}$, that is $C\left( h\right) =Cov\left( x_{t},x_{t+h}\right) $, which
is related to the spectral density function $f\left( \lambda \right) $ via
the expression 
\begin{equation*}
C\left( h\right) =\int_{\Pi ^{2}}f\left( \lambda \right) e^{-ih\cdot \lambda
}d\lambda \text{,}\quad h\in \mathbb{Z}^{2}\text{,}
\end{equation*}%
where $\Pi =\left( -\pi ,\pi \right] $. Henceforth the notation
\textquotedblleft $h\cdot \lambda $\textquotedblright\ means the inner
product of the vectors $h$ and $\lambda $. It is worth observing that we can
factorize $f\left( \lambda \right) $ as 
\begin{equation*}
f\left( \lambda \right) =\frac{\sigma _{\varepsilon }^{2}}{\left( 2\pi
\right) ^{2}}\left\vert \Psi \left( \lambda \right) \right\vert ^{2}=:\frac{%
\sigma _{\vartheta }^{2}}{\left( 2\pi \right) ^{2}}\left\vert \Upsilon
\left( \lambda \right) \right\vert ^{2}\text{, \ \ \ }\lambda \in \Pi ^{2}%
\text{,}
\end{equation*}%
where $\sigma _{\varepsilon }^{2}=E\varepsilon _{t}^{2}$ and $\sigma
_{\vartheta }^{2}=E\vartheta _{t}^{2}$, and%
\begin{equation*}
\Psi \left( \lambda \right) =\sum_{j\in \mathbb{Z}^{2}}\psi _{j}e^{-ij\cdot
\lambda };\text{ \ \ \ \ \ }\Upsilon \left( \lambda \right) =1+\sum_{0\prec
j}\zeta _{j}e^{-ij\cdot \lambda }\text{.}
\end{equation*}%
The latter displayed expressions indicate that either $\Psi \left( \lambda
\right) $ or $\Upsilon \left( \lambda \right) $ summarize the covariogram
structure of $\left\{ x_{t}\right\} _{t\in \mathbb{Z}^{2}}$.

When $d=1$ and the sequence $\left\{ x_{t}\right\} _{t\in \mathbb{Z}}$ is
purely nondeterministic we know, see Whittle $\left( 1961\text{, }%
p.26\right) $ or Brillinger $\left( 1981\text{, Theorem }3.8.4\right) $,
that the spectral density $f\left( \lambda \right) $ admits a representation 
\begin{equation*}
f\left( \lambda \right) =:\frac{\exp \left( -\alpha _{0}\right) }{2\pi }%
\left\vert A\left( \lambda \right) \right\vert ^{2}=\frac{1}{2\pi }\exp
\left\{ -\alpha _{0}-2\sum_{k=1}^{\infty }\alpha _{k}\cos \left( k\lambda
\right) \right\} \text{,}
\end{equation*}%
where by definition $A\left( \lambda \right) =:\exp \left\{
-\sum_{k=1}^{\infty }\alpha _{k}e^{ik\cdot \lambda }\right\} $. The latter
expression is referred to as the canonical factorization of the spectral
density function and is also known as Bloomfield's model. One important
consequence of the canonical factorization is that the sequence $\left\{
x_{t}\right\} _{t\in \mathbb{Z}}$ can be written as 
\begin{equation*}
x_{t}+\sum_{j=1}^{\infty }a_{j}x_{t-j}=\vartheta _{t}\text{,}
\end{equation*}%
where $\vartheta _{t}$ is a zero mean white noise sequence with finite
second moments and $a_{j}$ are the Fourier coefficients of $A\left( \lambda
\right) $, that is 
\begin{equation*}
a_{j}=\frac{1}{2\pi }\int_{-\pi }^{\pi }A\left( \lambda \right) e^{ij\cdot
\lambda }d\lambda ;\text{ \ \ \ \ }0<j,
\end{equation*}%
with $2\pi \exp \left( \alpha _{0}\right) =\sigma _{\vartheta }^{2}$, i.e.
the one-step prediction error. However, more importantly, denoting 
\begin{equation*}
B\left( \lambda \right) =A^{-1}\left( \lambda \right) =\exp \left\{
\sum_{k=1}^{\infty }\alpha _{k}e^{ik\cdot \lambda }\right\} \text{,}
\end{equation*}%
we have that its Fourier coefficients equal the coefficients $\zeta _{j}$ in 
$\left( \ref{uni_1}\right) $.

\cite{Whittle1954}, Section 6, signalled that a similar argument can be used
when $d>1$. However a formal and theoretical justification for a canonical
factorization of $f\left( \lambda \right) $ when $d>1$ was discussed in \cite%
{Korezlioglu1986}, see also \cite{Solo1986}. More specifically, they show
that the spectral density function of $\left\{ x_{t}\right\} _{t\in \mathbb{Z%
}^{2}}$ might be characterized using the representation 
\begin{equation}
f\left( \lambda \right) =:\frac{\exp \left( -\alpha _{0}\right) }{\left(
2\pi \right) ^{2}}\left\vert A\left( \lambda \right) \right\vert ^{2}=\frac{1%
}{\left( 2\pi \right) ^{2}}\exp \left\{ -\alpha _{0}-2\sum_{0\prec k}\alpha
_{k}\cos \left( k\cdot \lambda \right) \right\} \text{,}  \label{bloom_1}
\end{equation}%
where 
\begin{equation}
A\left( \lambda \right) =:\exp \left\{ -\sum_{0\prec k}\alpha _{k}e^{ik\cdot
\lambda }\right\} \text{,}  \label{arbloompa}
\end{equation}%
which is sometimes known as the \emph{Cepstrum} model by \cite{Solo1986},
who notes that if $0<f\left( \lambda \right) <M$ then the representation of
the spectral density in $\left( \ref{bloom_1}\right) $ or in $\left( \ref%
{arbloompa}\right) $ exists, see also \cite{mcelroy2014}. Note that the
coefficients $\alpha _{k}$ in $\left( \ref{arbloompa}\right) $ are the
Fourier coefficients of $\log \left( f\left( \lambda \right) \right) $, that
is 
\begin{equation}
\alpha _{k}=\frac{1}{2\pi ^{2}}\int_{\widetilde{\Pi }^{2}}\log \left(
f\left( \lambda \right) \right) \cos \left( k\cdot \lambda \right) d\lambda 
\text{, \ \ \ \ \ \ \ }0\prec k\text{ and }k=0\text{,}  \label{alpha_1}
\end{equation}%
where $\widetilde{\Pi }^{2}=\left[ 0,\pi \right] \times \Pi $, that is $%
\lambda \in \widetilde{\Pi }^{2}$ if $\lambda \left[ 1\right] \in \left[
0,\pi \right] $ and $\lambda \left[ 2\right] \in \Pi $. 
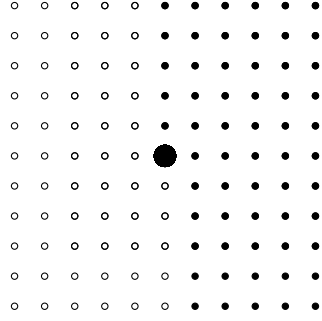
\begin{figure}[tbp]
\centering
\begin{tikzpicture}[scale=0.4]
    \foreach \i in {-5,...,5}
      \foreach \j in {-5,...,5}{
        \draw (\i,\j) circle(3pt);
        \draw (0,0) circle(10pt);
                          \fill[black] (0,0) circle(10pt);
        \ifnum \i > 0
          \fill[black] (\i,\j) circle(3pt);
        \fi
      };
          \foreach \i in {-3,...,5}
      \foreach \j in {-3,...,5}{
        \draw (\i,\j) circle(3pt);
        \ifnum \i =0 
        \ifnum \j>0
          \fill[black] (\i,\j) circle(3pt);
        \fi
        \fi
      };

  \end{tikzpicture}
\caption{Half-plane illustration for $d=2$. Circles form the half plane $
{\prec }0$ while solid dots form the half plane $0{\prec }$. The large
black solid dot marks the origin.}
\label{fig:hp_ill}
\end{figure}

As it is the case when $d=1$, there is a relationship between the
representation in $\left( \ref{uni_1}\right) $ and $\left( \ref{bloom_1}%
\right) /\left( \ref{arbloompa}\right) $, i.e. between the coefficients $%
\zeta _{j}$ and $\alpha _{k}$. So, it will be convenient to discuss the
relationship between the representations of the sequence $\left\{
x_{t}\right\} _{t\in \mathbb{Z}^{2}}$ in the \textquotedblleft
frequency\textquotedblright\ and \textquotedblleft space\textquotedblright\
domains. The link among these coefficients turns out to play a crucial role
in our prediction algorithm. For that purpose, consider the lexicographic
ordering given in $\left( \ref{lex_1}\right) $. Then, denoting the Fourier
coefficients of $A\left( \lambda \right) $ by 
\begin{equation}
a_{j}=\frac{1}{4\pi ^{2}}\int_{\Pi ^{2}}A\left( \lambda \right) e^{ij\cdot
\lambda }d\lambda ;\text{ \ \ \ \ }0\prec j,  \label{a_1}
\end{equation}%
and $a_{0}=1$, the sequence $\left\{ x_{t}\right\} _{t\in \mathbb{Z}^{2}}$
has a unilateral representation given by 
\begin{equation}
x_{t}+\sum_{0\prec j}a_{j}x_{t-j}=\vartheta _{t}\text{,}  \label{arunil}
\end{equation}%
where $\left\{ \vartheta _{t}\right\} _{t\in \mathbb{Z}^{2}}$ is the
sequence given in $\left( \ref{uni_1}\right) $. But also we have that the
coefficients $\zeta _{j}$ in $\left( \ref{uni_1}\right) $ are the Fourier
coefficients of $B\left( \lambda \right) =:A^{-1}\left( \lambda \right)
=\exp \left\{ \sum_{0\prec k}\alpha _{k}e^{-ik\cdot \lambda }\right\} $.
That is, 
\begin{equation*}
\zeta _{j}=\frac{1}{\left( 2\pi \right) ^{2}}\int_{\Pi ^{2}}B\left( \lambda
\right) e^{ij\cdot \lambda }d\lambda \text{, \ \ \ }0\prec j;\text{ \ \ \ }%
\zeta _{0}=1\text{,}
\end{equation*}%
see Section 1.2 of \cite{Korezlioglu1986}. The latter might be considered as
an extension of the canonical factorization given in \cite{Brillinger1981}
to the case $d>1$. However, one key aspect is that there is a direct link
between $\alpha _{k}$ and the coefficients of the Wold-type decomposition of
its autoregressive representation, that is $a_{j}/\zeta _{j}$ and $\alpha
_{k}$. This observation will be important for our prediction methodology in
the next section.

\section{\textbf{Prediction algorithm }}

\label{sec:FE_pred}

The purpose of the section is to present and examine a prediction algorithm,
extending the methodology in \cite{Bhansali1974} or \cite{Hidalgo2002}, to
the case when $d=2$. Similar to the aforementioned work, a key component of
the methodology will be based on the canonical factorization of the spectral
density in $\left( \ref{bloom_1}\right) $. Due to the rather unusual
notation in this paper, we have decided to collate it at this stage for
convenience. Given two vectors $a$ and $b$, $a\geq \left( \leq \right) b$
means that $a\left[ \ell \right] \geq \left( \leq \right) b\left[ \ell %
\right] $ for all $\ell =1,2$. Denote%
\begin{equation*}
\Pi _{n}^{2}=\left\{ \lambda _{k\left[ \ell \right] }=\frac{2\pi k\left[
\ell \right] }{n\left[ \ell \right] }\text{, \ }\ k\left[ \ell \right]
=0,\pm 1,...,\pm \tilde{n}\left[ \ell \right] =:\frac{n\left[ \ell \right] }{%
2}\text{, \ \ }\ell =1,2\right\} \text{,}
\end{equation*}%
where $\lambda _{k}=\left( \lambda _{k\left[ 1\right] },\lambda _{k\left[ 2%
\right] }\right) $ are the Fourier frequencies and $\widetilde{\Pi }%
_{n}^{2}=\left\{ \lambda _{k}\in \Pi _{n}^{2}:\lambda _{k\left[ 1\right]
}>0\right\} $. Finally, we denote 
\begin{eqnarray}
\int_{\lambda \preceq \pi }^{+} &=&\int_{\lambda \left[ 1\right] =0}^{\pi
}\int_{\lambda \left[ 2\right] =-\pi }^{\pi }\text{; \ \ \ \ }\int_{\lambda
\preceq \pi }^{-}=\int_{\lambda \left[ 1\right] =-\pi }^{0}\int_{\lambda %
\left[ 2\right] =-\pi }^{\pi };  \label{notc} \\
\int_{a\leq \lambda \leq b} &=&\int_{\lambda \left[ 1\right] =a\left[ 1%
\right] }^{b\left[ 1\right] }\int_{\lambda \left[ 2\right] =b\left[ 2\right]
}^{b\left[ 2\right] }\text{.}  \notag
\end{eqnarray}%
Similarly, we denote 
\begin{eqnarray}
\left. \sum_{j\preceq J}\right. ^{+}c_{j} &=&\sum_{j\left[ 2\right] =1}^{J%
\left[ 2\right] }c_{0,j\left[ 2\right] }+\sum_{j\left[ 1\right] =1}^{J\left[
1\right] }\sum_{j\left[ 2\right] =1-J\left[ 2\right] }^{J\left[ 2\right]
}c_{j\left[ 1\right] ,j\left[ 2\right] };~\ \ \ \   \notag \\
\left. \sum_{j\preceq J}\right. ^{-}c_{j} &=&\sum_{j\left[ 2\right] =1-J%
\left[ 2\right] }^{0}c_{0,j\left[ 2\right] }+\sum_{j\left[ 1\right] =1-J%
\left[ 1\right] }^{0}\sum_{j\left[ 2\right] =1-J\left[ 2\right] }^{J\left[ 2%
\right] }c_{j\left[ 1\right] ,j\left[ 2\right] };  \label{notd} \\
\sum_{a\leq t\leq b} &=&\sum_{t\left[ 1\right] =a\left[ 1\right] }^{b\left[ 1%
\right] }\sum_{t\left[ 2\right] =a\left[ 2\right] }^{b\left[ 2\right] }\text{%
,}  \notag
\end{eqnarray}%
where we are using the convention that for any $k\in \mathbb{Z}^{2}$, we
write $d_{k}$ as 
\begin{equation*}
d_{k}=d_{k\left[ 1\right] ,k\left[ 2\right] }\text{.}
\end{equation*}%
Observe that $\sum_{j\preceq J}^{+}+\sum_{j\preceq J}^{-}=\sum_{-J<j\leq J}$%
, and likewise $\int_{\lambda \preceq \pi }^{+}+\int_{\lambda \preceq \pi
}^{-}=\int_{\lambda \in \Pi ^{2}}$.

Before we describe our prediction algorithm, we shall introduce our set of
regularity conditions.

\begin{description}
\item[\textbf{Condition C1}] $\left( \mathbf{a}\right) $ $\left\{ \vartheta
_{t}\right\} _{t\in \mathbb{Z}^{2}}$ in $\left( \ref{uni_1}\right) $ is a
zero mean white noise sequence of random variables with variance $\sigma
_{\vartheta }^{2}$ and finite $4th$ moments, with $\kappa _{4,\vartheta }$
denoting the fourth cumulant of $\vartheta _{t}$.

$\left( \mathbf{b}\right) $ The unilateral \emph{Moving Average}
representation of $\left\{ x_{t}\right\} _{t\in \mathbb{Z}^{2}}$ in $\left( %
\ref{uni_1}\right) $ can be written (or it has a representation) as a
unilateral \emph{Autoregressive} model\emph{\ }%
\begin{equation}
x_{t}+\sum_{0\prec j}a_{j}x_{t-j}=\vartheta _{t}\text{.}  \label{SAR}
\end{equation}

$\left( \mathbf{c}\right) $ The coefficients in $\zeta _{j}$ in $\left( \ref%
{uni_1}\right) $ satisfy 
\begin{equation*}
\sum_{0\prec j}\left\{ \sum_{\ell =1}^{2}j^{4}\left[ \ell \right] \right\}
\left\vert \zeta _{j}\right\vert <\infty \text{.}
\end{equation*}

\item[\textbf{Condition C2}] $n=\left( n\left[ 1\right] ,n\left[ 2\right]
\right) $ satisfies that $n\left[ 1\right] \asymp n\left[ 2\right] $ where
\textquotedblleft $a\asymp b$\textquotedblright\ means that $K^{-1}\leq
a/b\leq K$ for some finite positive constant $K$.
\end{description}

We now comment on Conditions $C1$ and $C2$. First, Condition $C2$ can be
generalized to allow for different rates of convergence to zero of $n^{-1}%
\left[ \ell \right] $, $\ell =1,2$. However, for notational simplicity, we
prefer to keep it as it stands. Condition $C1$ could have been written in
terms of the multilateral representation in $\left( \ref{a1}\right) $.
However since the prediction employs the representation in $\left( \ref{SAR}%
\right) $ or $\left( \ref{uni_1}\right) $, we have opted to write $C1$ as it
stands. Part $\left( a\right) $ of Condition $C1$ seems to be a minimal
condition for our results below to hold true. Sufficient regularity
conditions required for the validity of the expansion in $\left( \ref{SAR}%
\right) $ is $\Upsilon \left( z\right) $ be nonzero for any $z\left[ \ell %
\right] $, $\ell =1,2$. The latter condition guarantees that $f\left(
\lambda \right) >0$ for all $\lambda \in \widetilde{\Pi }^{2}$. Part $\left( 
\mathbf{c}\right) $ entails that the spectral density $f\left( \lambda
\right) $ is $4$ times continuously differentiable. This is needed if one
wants to achieve a similar rate of approximation of sums by their integrals
when $d=1$ and the function is twice continuously differentiable. Indeed
whereas when $d=1$, we have that 
\begin{equation*}
\frac{1}{n}\sum_{i=1}^{n}g\left( \frac{i}{n}\right) -\int_{0}^{1}g\left(
x\right) =\frac{1}{n}\left( g\left( 0\right) -g\left( 1\right) \right)
+O\left( n^{-2}\right) ,
\end{equation*}%
with two continuous derivatives for $g\left( x\right) $, to have a
\textquotedblleft similar\textquotedblright\ result when $d=2$ one needs $%
g(x)$ to be $4$ times continuously differentiable. See Lemma \ref{bias_2} in
the appendix for some extra insight.

We now discuss the methodology to predict the\ value of $x_{t}$ at an
unobserved location without imposing any specific parametric model for $%
f\left( \lambda \right) $. In addition, as a by-product, we provide a simple
estimator of the coefficients $\zeta _{j}$ or $a_{j}$. First, $A\left(
\lambda \right) $ and expression $\left( \ref{alpha_1}\right) $ suggest that
to compute an estimator of the coefficients $\alpha _{j}$ and/or $a_{j}$, it
suffices to obtain an estimator of $f\left( \lambda \right) $. To that end,
for a generic sequence $\left\{ v_{t}\right\} _{t=1}^{n}$, we shall define
the \emph{discrete Fourier transform}, $DFT$, as 
\begin{equation*}
w_{v}\left( \lambda \right) =\frac{1}{\mathbf{n}^{1/2}}\sum_{1\leq t\leq
n}v_{t}e^{-it\cdot \lambda },
\end{equation*}%
and the \emph{periodogram }as 
\begin{equation*}
I_{v}\left( \lambda \right) =\frac{1}{\left( 2\pi \right) ^{2}}\left\vert
w_{v}\left( \lambda \right) \right\vert ^{2}\text{; \ \ \ \ \ }\lambda \in 
\widetilde{\Pi }^{2}\text{,}
\end{equation*}%
where, in what follows, we use the notation that for any $g=\left( g\left[ 1%
\right] ,g\left[ 2\right] \right) $, 
\begin{equation}
\mathbf{g}=g\left[ 1\right] g\left[ 2\right] \text{.}  \label{gblack}
\end{equation}%
In real applications, in order to make use of the \emph{fast Fourier
transform}, the \emph{periodogram} will be evaluated at the Fourier
frequencies $\lambda _{k}$.

However as noted by \cite{Guyon1982}, due to non-negligible end effects (the
edge effect), the bias of the periodogram does not converge to zero fast
enough when $d>1$. We therefore proceed as in \cite{Dahlhaus1987}, and
employ the \emph{tapered periodogram} defined as 
\begin{equation}
I_{v}^{T}\left( \lambda _{j}\right) =\frac{1}{\left( 2\pi \right) ^{2}}%
\left\vert w_{v}^{T}\left( \lambda _{j}\right) \right\vert ^{2}\text{; \ \ \
\ }w_{v}^{T}\left( \lambda _{j}\right) =\frac{1}{\left( \sum_{1\leq t\leq
n}h_{t}^{2}\right) ^{1/2}}\sum_{1\leq t\leq n}h_{t}v_{t}e^{it\cdot \lambda
_{j}}\text{,}  \label{pert}
\end{equation}%
where $w_{v}^{T}\left( \lambda _{j}\right) $ denotes the taper discrete
Fourier transform, $DFT$. One common taper is the \emph{cosine-bell (or
Hanning)} function, which is defined as 
\begin{equation}
h_{t}=\frac{1}{4}h_{t\left[ 1\right] }h_{t\left[ 2\right] }\text{; \ \ \ \ }%
h_{t\left[ \ell \right] }=\left( 1-\cos \left( \frac{2\pi t\left[ \ell %
\right] }{n\left[ \ell \right] }\right) \right) \text{,}  \label{cos-bell}
\end{equation}%
see \cite{Brillinger1981}. It is worth observing the \emph{cosine-bell }%
taper \emph{DFT} is related to $w_{v}\left( \lambda \right) $ by the
equality 
\begin{equation}
w_{v}^{T}\left( \lambda _{j}\right) =\frac{1}{6}\prod\limits_{\ell =1}^{2}%
\left[ -w_{v}\left( \lambda _{j\left[ \ell \right] -1}\right) +2w_{v}\left(
\lambda _{j\left[ \ell \right] }\right) -w_{v}\left( \lambda _{j\left[ \ell %
\right] +1}\right) \right] \text{.}  \label{dft}
\end{equation}%
In this paper we shall explicitly consider the \emph{cosine-bell}, although
the same results follow employing other taper functions such as \emph{Parzen
or Kolmogorov} tapers \citep{Brillinger1981}. \ This is formalized in the
next condition.

\begin{description}
\item[\textbf{Condition C3}] $\left\{ h_{t}\right\} _{t=1}^{n}$ is the
cosine-bell taper function in $\left( \ref{cos-bell}\right) $.
\end{description}

Using notation in $\left( \ref{gblack}\right) $, we shall estimate $f\left(
\lambda \right) $ by the average tapered periodogram 
\begin{equation}
\widehat{f}\left( \lambda \right) =\frac{1}{4\mathbf{m}}\sum_{-m<\ell \leq
m}I_{x}^{T}\left( \lambda +\lambda _{\ell }\right) \text{,}  \label{fest}
\end{equation}%
where $m\left[ \ell \right] /n\left[ \ell \right] +m\left[ \ell \right]
^{-1}=o\left( 1\right) $,$\ $for $\ell =1,2$. Next, we denote $\widetilde{%
\lambda }_{k}=\left( \widetilde{\lambda }_{k\left[ 1\right] },\widetilde{%
\lambda }_{k\left[ 2\right] }\right) ^{\prime }$, for $k\left[ 1\right]
=0,1,...,M\left[ 1\right] =:\tilde{n}\left[ 1\right] /m\left[ 1\right] $ and 
$k\left[ 2\right] =0,\pm 1,...,\pm M\left[ 2\right] =:\tilde{n}\left[ 2%
\right] /m\left[ 2\right] $, where 
\begin{equation*}
\widetilde{\lambda }_{k\left[ \ell \right] }=\frac{\pi k\left[ \ell \right] 
}{M\left[ \ell \right] }\text{; \ \ \ }\ell =1,2\text{.}
\end{equation*}%
Bearing in mind $\left( \ref{notd}\right) $, denoting $\mathcal{M=}\left\{
j:~\left( 0\prec j;j=0\right) \text{ }\wedge \left( -M<j\leq M\right)
\right\} $ and abbreviating $\phi \left( \widetilde{\lambda }_{k}\right) $
by $\phi _{k}$ for a generic function $\phi \left( \lambda \right) $, we
estimate the coefficients $a_{j}$, $j=1,...,M$, as%
\begin{equation}
\widehat{a}_{j}=\frac{1}{4\mathbf{M}}\sum_{-M<k\leq M}\widehat{A}%
_{k}e^{ij\cdot \widetilde{\lambda }_{k}}\text{, \ \ \ }\ j\in \mathcal{M};
\label{ahat_j}
\end{equation}%
\begin{equation*}
\widehat{A}_{k}=\overline{\widehat{A}}_{-k}=\exp \left\{ -\left.
\sum_{j\preceq M}\right. ^{+}\widehat{\alpha }_{j}e^{-ij\cdot \widetilde{%
\lambda }_{k}}\right\} \text{, \ \ \ \ \ \ \ }k\in \mathcal{M\cup }\left\{
0\right\} ;
\end{equation*}%
\begin{equation}
\widehat{\alpha }_{j}=\frac{1}{2\mathbf{M}}\left. \sum_{k\preceq M}\right.
^{+}\cos \left( j\cdot \widetilde{\lambda }_{k}\right) \log \widehat{f}_{k}%
\text{, \ \ \ \ \ \ \ }j\in \mathcal{M}.  \label{cr_1}
\end{equation}

It is also worth defining the quantities $\left( \ref{ahat_j}\right) $ and $%
\left( \ref{cr_1}\right) $ when $\widehat{f}\left( \lambda \right) $ is
replaced by $f\left( \lambda \right) $, that is 
\begin{equation*}
\widetilde{f}\left( \lambda \right) =\frac{1}{4\mathbf{m}}\sum_{-m<k\leq
m}f\left( \lambda +\lambda _{k}\right) \text{.}
\end{equation*}%
That is, 
\begin{eqnarray}
\widetilde{a}_{j,n} &=&\frac{1}{4\mathbf{M}}\sum_{-M<k\leq M}\widetilde{A}%
_{k,n}e^{ij\cdot \widetilde{\lambda }_{k}}\text{ \ \ \ \ \ \ \ \ \ \ \ \ \ \
\ \ \ \ \ \ \ \ \ \ \ \ \ \ \ \ }j\in \mathcal{M}  \notag \\
\widetilde{A}_{k,n} &=&\overline{\widetilde{A}}_{k,n}=\exp \left\{ -\left.
\sum_{j\preceq M}\right. ^{+}\widetilde{\alpha }_{j,n}e^{-ij\cdot \widetilde{%
\lambda }_{k}}\right\} \text{ \ \ \ \ \ \ }k\in \mathcal{M\cup }\left\{
0\right\}  \label{an_j} \\
\widetilde{\alpha }_{j,n} &=&\frac{1}{2\mathbf{M}}\left. \sum_{k\preceq
M}\right. ^{+}\cos \left( j\cdot \widetilde{\lambda }_{k}\right) \log 
\widetilde{f}_{k}\text{ \ \ \ \ \ \ \ \ \ \ \ }j\in \mathcal{M\cup }\left\{
0\right\}  \notag
\end{eqnarray}%
and also we denote%
\begin{eqnarray}
a_{j,n} &=&\frac{1}{4\mathbf{M}}\sum_{-M<k\leq M}A_{k,n}e^{ij\cdot 
\widetilde{\lambda }_{k}}\text{ \ \ \ \ \ \ \ \ \ \ \ \ \ \ \ \ \ \ \ \ \ \
\ \ \ \ \ \ \ \ \ }j\in \mathcal{M}  \notag \\
A_{k,n} &=&\overline{A}_{k,n}=\exp \left\{ -\left. \sum_{j\preceq M}\right.
^{+}\alpha _{j,n}e^{-ij\cdot \widetilde{\lambda }_{k}}\right\} \text{ \ \ \
\ \ \ }k\in \mathcal{M\cup }\left\{ 0\right\}  \label{crn_1} \\
\alpha _{j,n} &=&\frac{1}{2\mathbf{M}}\left. \sum_{k\preceq M}\right.
^{+}\cos \left( j\cdot \widetilde{\lambda }_{k}\right) \log f_{k}\text{ \ \
\ \ \ \ \ \ \ \ \ \ }j\in \mathcal{M\cup }\left\{ 0\right\} \text{.}  \notag
\end{eqnarray}

We shall now begin describing how we can predict a value $x_{t}$ at the
location $s=\left( s\left[ 1\right] ,s\left[ 2\right] \right) $ such that $%
1\leq s\left[ 1\right] \leq n\left[ 1\right] $ and $1\leq s\left[ 2\right]
\leq n\left[ 2\right] $. For instance, we wish to predict the unobserved
value\ $x_{s}$%
\begin{equation*}
\ 
\begin{array}{c}
\bullet \\ 
\bullet \\ 
\bullet \\ 
\bullet \\ 
\bullet%
\end{array}%
\begin{array}{c}
\bullet \\ 
\bullet \\ 
\bullet \\ 
\bullet \\ 
\bullet%
\end{array}%
\begin{array}{c}
\bullet \\ 
\bullet \\ 
\bullet \\ 
\bullet \\ 
\bullet%
\end{array}%
\begin{array}{c}
\bullet \\ 
\bullet \\ 
\bullet \\ 
\bullet \\ 
\bullet%
\end{array}%
\begin{array}{c}
\bullet \\ 
\bullet \\ 
\bullet \\ 
\bullet \\ 
\bullet%
\end{array}%
\begin{array}{c}
\bullet \\ 
\bullet \\ 
\bullet \\ 
\bullet \\ 
\bullet%
\end{array}%
\begin{array}{c}
\bullet \\ 
\bullet \\ 
\left( s\left[ 1\right] ,s\left[ 2\right] \right) ? \\ 
\bullet \\ 
\bullet%
\end{array}%
\begin{array}{c}
\bullet \\ 
\bullet \\ 
\bullet \\ 
\bullet \\ 
\bullet%
\end{array}%
\ \ \ \ \ \ \ \ \text{or}\ \ \ 
\begin{array}{c}
\bullet \\ 
\bullet \\ 
\bullet \\ 
\bullet \\ 
\bullet%
\end{array}%
\begin{array}{c}
\bullet \\ 
\bullet \\ 
\bullet \\ 
\bullet \\ 
\bullet%
\end{array}%
\begin{array}{c}
\bullet \\ 
\bullet \\ 
\bullet \\ 
\bullet \\ 
\bullet%
\end{array}%
\begin{array}{c}
\bullet \\ 
\bullet \\ 
\bullet \\ 
\bullet \\ 
\bullet%
\end{array}%
\begin{array}{c}
\bullet \\ 
\bullet \\ 
\bullet \\ 
\bullet \\ 
\bullet%
\end{array}%
\begin{array}{c}
\bullet \\ 
\bullet \\ 
\left( s\left[ 1\right] ,s\left[ 2\right] \right) ? \\ 
\bullet \\ 
\bullet%
\end{array}%
\end{equation*}%
Now, the location of $s$ suggests that a convenient unilateral
representation of $x_{t}$ appears to be%
\begin{equation}
x_{t}=-\sum_{k\left[ 2\right] =1}^{\infty }a_{0,k\left[ 2\right] }~x_{t\left[
1\right] ,t\left[ 2\right] -k\left[ 2\right] }-\sum_{k\left[ 1\right]
=1}^{\infty }\sum_{k\left[ 2\right] =-\infty }^{\infty
}a_{k}~x_{t-k}+\vartheta _{t}\text{,}  \label{unil_1}
\end{equation}%
which comes from the lexicographic ordering in $\left( \ref{lex_1}\right) $.
Since we need to estimate the coefficients $a_{k}$, the prediction will then
become 
\begin{equation}
\widehat{x}_{s\left[ 1\right] ,s\left[ 2\right] }=-\sum_{k\left[ 2\right]
=1}^{M\left[ 2\right] }\widehat{a}_{0,k\left[ 2\right] }~x_{s\left[ 1\right]
,s\left[ 2\right] -k\left[ 2\right] }-\sum_{k\left[ 1\right] =1}^{M\left[ 1%
\right] }\sum_{k\left[ 2\right] =1-M\left[ 2\right] }^{M\left[ 2\right] }%
\widehat{a}_{k}~x_{s-k}\text{,}  \label{1}
\end{equation}%
where $\widehat{a}_{k}~x_{s-k}=:\widehat{a}_{k\left[ 1\right] ,k\left[ 2%
\right] }~x_{s\left[ 1\right] -k\left[ 1\right] ,s\left[ 2\right] -k\left[ 2%
\right] }$. However, it may be very plausible that the value$\ $of $M$ is
such that we may not observe the process at some of the locations employed
to compute $\left( \ref{1}\right) $. That is, consider the situation where
we want to predict $x_{s}$%
\begin{equation*}
\begin{array}{c}
\bullet \\ 
\bullet \\ 
\bullet \\ 
\bullet \\ 
\bullet%
\end{array}%
\begin{array}{c}
\bullet \\ 
\bullet \\ 
\bullet \\ 
\bullet \\ 
\bullet%
\end{array}%
\begin{array}{c}
\bullet \\ 
\bullet \\ 
\bullet \\ 
\bullet \\ 
\bullet%
\end{array}%
\begin{array}{c}
\bullet \\ 
\bullet \\ 
\bullet \\ 
\bullet \\ 
\bullet%
\end{array}%
\begin{array}{c}
\bullet \\ 
\bullet \\ 
\bullet \\ 
\bullet \\ 
\bullet%
\end{array}%
\begin{array}{c}
\bullet \\ 
\bullet \\ 
\bullet \\ 
\bullet \\ 
\bullet%
\end{array}%
\begin{array}{c}
\left( s\left[ 1\right] ,s\left[ 2\right] \right) ? \\ 
\bullet \\ 
\bullet \\ 
\bullet \\ 
\bullet%
\end{array}%
\begin{array}{c}
\bullet \\ 
\bullet \\ 
\bullet \\ 
\bullet \\ 
\bullet%
\end{array}%
\end{equation*}

In this case we observe that to compute $\left( \ref{1}\right) $, we first
need to obtain a predictor of values of $x_{s-k}$ when say $k\left[ 1\right]
=1$ and $k\left[ 2\right] <0$, since $x_{s-k}$ is not observed at those
locations, which in its computation needs predictors of the relevant values
themselves. See $\left( \ref{pre_1}\right) $ for more exact details.
However, in this case one can avoid this extra computational burden. Indeed,
this is so as the relative location $\left( s\left[ 1\right] ,s\left[ 2%
\right] \right) $ suggests that the practitioner might have employed the
Wold-type representation%
\begin{equation}
x_{t}=-\sum_{k\left[ 1\right] =1}^{\infty }a_{k\left[ 1\right] ,0}~x_{t\left[
1\right] -k\left[ 1\right] ,t\left[ 2\right] }-\sum_{k\left[ 2\right]
=1}^{\infty }\sum_{k\left[ 1\right] =-\infty }^{\infty
}a_{k}~x_{t-k}+\vartheta _{t}  \label{unil_2}
\end{equation}%
which can be regarded as induced by the lexicographic ordering\ 
\begin{equation}
j\prec k\Leftrightarrow \left( j\left[ 2\right] <k\left[ 2\right] \right) 
\text{ or }\left( j\left[ 2\right] =k\left[ 2\right] \vee j\left[ 1\right] <k%
\left[ 1\right] \right) \text{.}  \label{lex_2}
\end{equation}%
Note that the lexicographic ordering $\left( \ref{lex_2}\right) $ is as that
in $\left( \ref{lex_1}\right) $ but swapping $j\left[ 2\right] $ for $j\left[
1\right] $. From here, we proceed as with $\left( \ref{1}\right) $ but with
the \textquotedblleft coordinates\textquotedblright\ $\left[ 2\right] $ and $%
\left[ 1\right] $ changing their roles.

Finally, consider the case where location we wish to predict $x_{s}$ is $%
\left( n\left[ 1\right] +1,s\left[ 2\right] \right) $. That is,%
\begin{equation*}
\begin{array}{c}
. \\ 
\bullet \\ 
\bullet \\ 
\bullet \\ 
\bullet%
\end{array}%
\begin{array}{c}
. \\ 
\bullet \\ 
\bullet \\ 
\bullet \\ 
\bullet%
\end{array}%
\begin{array}{c}
. \\ 
\bullet \\ 
\bullet \\ 
\bullet \\ 
\bullet%
\end{array}%
\begin{array}{c}
. \\ 
\bullet \\ 
\bullet \\ 
\bullet \\ 
\bullet%
\end{array}%
\begin{array}{c}
. \\ 
\bullet \\ 
\bullet \\ 
\bullet \\ 
\bullet%
\end{array}%
\begin{array}{c}
. \\ 
\bullet \\ 
\bullet \\ 
\bullet \\ 
\bullet%
\end{array}%
\begin{array}{c}
. \\ 
. \\ 
. \\ 
\left( n\left[ 1\right] +1,s\left[ 2\right] \right) \\ 
.%
\end{array}%
\end{equation*}

Now, the location of $s=:\left( n\left[ 1\right] +1,s\left[ 2\right] \right) 
$ suggests that the more convenient representation of $x_{s}$ appears to be
that in $\left( \ref{unil_1}\right) $ which comes from the lexicographic
ordering in $\left( \ref{lex_1}\right) $, and hence our prediction is given
in $\left( \ref{1}\right) $. That is, since we need to estimate the
coefficients $a_{k}$, the prediction will then become%
\begin{equation}
\widehat{x}_{n\left[ 1\right] +1,s\left[ 2\right] }=-\sum_{k\left[ 2\right]
=1}^{M\left[ 2\right] }\widehat{a}_{0,k\left[ 2\right] }~x_{n\left[ 1\right]
+1,s\left[ 2\right] -k\left[ 2\right] }-\sum_{k\left[ 1\right] =1}^{M\left[ 1%
\right] }\sum_{k\left[ 2\right] =1-M\left[ 2\right] }^{M\left[ 2\right] }%
\widehat{a}_{k}~x_{s-k}\text{.}  \label{1a}
\end{equation}%
However to compute the prediction we also need to replace the unobserved $%
x_{s}$ by its prediction. As with \textquotedblleft
standard\textquotedblright\ time series when we wish to predict beyond $1$
period ahead, this is done by recursion, that is we make use of formula $%
\left( \ref{1}\right) $ starting say from the value $x_{n\left[ 1\right] +1,s%
\left[ 2\right] -M\left[ 2\right] }$. Once we have \textquotedblleft
predicted\textquotedblright\ the value for this observation, we then predict 
$x_{n\left[ 1\right] +1,s\left[ 2\right] -M\left[ 2\right] +1}$ and so on.
For instance, for any $r\left[ 1\right] =0,...,r$ and $r\left[ 2\right]
=0,...,r=\min \left\{ n\left[ 2\right] /8;M\left[ 2\right] \right\} $, 
\begin{eqnarray}
\widehat{x}_{t\left[ 1\right] -r\left[ 1\right] ,t\left[ 2\right] -r\left[ 2%
\right] } &=&-\sum_{k\left[ 2\right] =1}^{M\left[ 2\right] }\widehat{a}_{0,k%
\left[ 2\right] }~\widehat{x}_{t\left[ 1\right] -r\left[ 1\right] ,t\left[ 2%
\right] -r\left[ 2\right] -k\left[ 2\right] }  \label{pre_1} \\
&&-\sum_{k\left[ 1\right] =1}^{M\left[ 1\right] }\sum_{k\left[ 2\right] =1-M%
\left[ 2\right] }^{M\left[ 2\right] }\widehat{a}_{k\left[ 1\right] ,k\left[ 2%
\right] }~\widehat{x}_{t\left[ 1\right] -r\left[ 1\right] -k\left[ 1\right]
,t\left[ 2\right] -r\left[ 2\right] -k\left[ 2\right] }\text{,}  \notag
\end{eqnarray}%
where we take the convention that $\widehat{x}_{s}=x_{s}$ if the location
were observed and $=:0$ when $s\left[ 2\right] <-r$ or $\left\{ s\left[ 2%
\right] <0\text{ }\wedge s\left[ 1\right] <n\left[ 1\right] -r\right\} $.
Finally, if we were interested to predict $x_{t}$ at the unobserved location 
$\left( t\left[ 1\right] ,n\left[ 2\right] +1\right) $, then it suggests to
employ the lexicographic ordering in $\left( \ref{lex_2}\right) $ and hence
the representation given in $\left( \ref{unil_2}\right) $, and then we would
proceed as above but again with the \textquotedblleft
coordinates\textquotedblright\ $\left[ 2\right] $ and $\left[ 1\right] $
changing their roles.

Before we examine the statistical properties of $\widehat{x}_{t}$ in $\left( %
\ref{1}\right) $ or $\left( \ref{1a}\right) $, we shall look at those of $%
\widehat{\alpha }_{j}$ or $\widehat{A}_{j}$. For that purpose, denote 
\begin{eqnarray}
\delta _{j} &:&=1\text{ \ if \ }j=0\text{ \ and}:=0\text{ \ otherwise}
\label{not_th} \\
\phi _{k} &=&\left\{ 
\begin{array}{c}
i\frac{1-\cos \left( k\pi \right) }{k\pi },\text{ \ \ \ \ \ if \ \ }k\in 
\mathbb{N}^{+} \\ 
1,\text{\ \ \ \ \ \ \ \ \ \ \ \ \ \ \ \ if \ \ \ \ \ }k=0\text{.}%
\end{array}%
\right.   \notag
\end{eqnarray}%
Also, denote $\left\{ \boldsymbol{\xi }_{j}\right\} _{j}$ the Fourier
coefficients of $g\left( \lambda \right) $ \ given by%
\begin{eqnarray}
g\left( \lambda \right)  &=&\frac{1}{6}\left( f_{11}\left( \lambda \right)
+f_{22}\left( \lambda \right) \right) \text{; \ \ \ \ }\lambda \in \Pi ^{2}
\label{g_1} \\
f_{\ell _{1}\ell _{2}}\left( \lambda \right)  &=&\frac{\partial ^{2}}{%
\partial \lambda _{\left[ \ell _{1}\right] }\partial \lambda _{\left[ \ell
_{2}\right] }}f\left( \lambda \right) ;\text{ \ }\ell _{1},\ell _{2}=1,2%
\text{.}  \notag
\end{eqnarray}%
Notice that Condition $C1$ implies that $g\left( \lambda \right) $ is twice
continuous differentiable, so that $\left\{ \boldsymbol{\xi }_{j}\right\}
_{j}$ is summable.

We introduce one extra condition relating the rate of increase of $m\left[
\ell \right] $ with respect of $n\left[ \ell \right] $.

\begin{description}
\item[\textbf{Condition C4}] $n\left[ \ell \right] ,m\left[ \ell \right]
\rightarrow \infty $, for $\ell =1,2$, such that%
\begin{equation*}
\frac{n^{3}\left[ \ell \right] }{m^{4}\left[ \ell \right] }+\frac{m\left[
\ell \right] }{n\left[ \ell \right] }\rightarrow 0~\ \ \ \ \ell =1,2\text{.}
\end{equation*}
\end{description}

\begin{theorem}
Under $C1-C4$, for any finite integer $J$, we have that 
\begin{equation*}
\left( \mathbf{a}\right) \text{ \ \ }\mathbf{n}^{1/2}\left( \widehat{\alpha }%
_{j}-\widetilde{\alpha }_{j,n}\right) _{j=1}^{J}\overset{d}{\rightarrow }%
\mathcal{N}\left( 0,\Omega _{\alpha }\right) \text{,}
\end{equation*}%
\begin{equation*}
\left( \mathbf{b}\right) \text{\ \ \ \ \ }\widetilde{\alpha }_{j,n}-\alpha
_{j,n}=O\left( \mathbf{M}^{-1}\mathbf{\xi }_{j}+\mathbf{M}^{2}\right) \text{%
, \ \ \ }j=1,...,J\text{,}
\end{equation*}%
where$\ \Omega _{a}$ is a diagonal matrix whose $(j,j)$-th element is $%
1+\left( 1+\kappa _{4,\vartheta }\right) \delta _{j}$.
\end{theorem}

\begin{description}
\item[\emph{Remark}] Because $\sigma _{\vartheta }^{2}=2\pi \exp \left(
\alpha _{0}\right) $, we have that $\widehat{\sigma }_{\vartheta }^{2}=:2\pi
\exp \left( \widehat{\alpha }_{0}\right) $ is a consistent estimator of $%
\sigma _{\vartheta }^{2}$. Indeed, by standard delta methods, the proof
follows using Theorem 1 and that Lemma \ref{bias_2} implies that $\alpha
_{0,n}-\alpha _{0}=O\left( \mathbf{M}^{-1/2}\right) $.
\end{description}

\begin{theorem}
Under $C1-C4$, for any finite integer $J$, we have that%
\begin{equation*}
\left( \mathbf{a}\right) \text{ \ \ \ \ \ \ \ \ \ \ \ \ \ \ \ \ \ \ \ \ \ \
\ }\mathbf{m}^{1/2}\left( \widehat{A}_{j}-\widetilde{A}_{j,n}\right)
_{j=1}^{J}\overset{d}{\rightarrow }\mathcal{N}^{c}\left( 0,\Omega
_{A}\right) \text{,}
\end{equation*}
\begin{equation*}
\left( \mathbf{b}\right) \text{ \ }\widetilde{A}_{j,n}-A_{j,n}=\frac{1}{%
\mathbf{M}}g_{j}A_{j,n}+o\left( \mathbf{m}^{-1/2}\right) \text{, \ \ \ }%
j=1,...,J\text{,}
\end{equation*}%
where $g_{j}=g\left( \widetilde{\lambda }_{j}\right) $ is given in $\left( %
\ref{g_1}\right) $ and $\mathcal{N}^{c}\left( 0,\Omega _{A}\right) $ denotes
a complex normal random variable with the $(j_{1},j_{2})$-th element of $%
\Omega _{A}$ given by 
\begin{equation*}
\Omega _{A,j_{1}j_{2}}=2\left( \delta _{j_{1}\left[ 1\right] -j_{2}\left[ 1%
\right] }+2^{-1}\phi _{j_{1}\left[ 1\right] }\phi _{j_{2}\left[ 1\right]
}-i\phi _{j_{1}\left[ 1\right] -j_{2}\left[ 1\right] }\right) \delta _{j_{1}%
\left[ 2\right] \pm j_{2}\left[ 2\right] }A_{j_{1}}\overline{A}_{j_{2}}\text{%
.}
\end{equation*}
\end{theorem}

We shall now denote $a_{\upsilon }=0$ if $\upsilon \prec 0$.

\begin{theorem}
Under $C1-C4$, for any finite integer $J$, we have that 
\begin{eqnarray*}
&&\left( \mathbf{a}\right) \text{ \ }\mathbf{n}^{1/2}\left( \widehat{a}_{j}-%
\widetilde{a}_{j,n}\right) _{j=1}^{J}\overset{d}{\rightarrow }\mathcal{N}%
\left( 0,\Omega _{a}\right) \text{,} \\
&&\left( \mathbf{b}\right) \text{ \ }\mathbf{n}^{1/2}\left( \widetilde{a}%
_{j,n}-a_{j,n}\right) \overset{}{\rightarrow }0\text{.}
\end{eqnarray*}%
where$\ $a typical element $\left( j_{1},j_{2}\right) $ of $\Omega _{a}$,
with $j_{1}\preceq j_{2}$, is $\sum_{0\preceq k}a_{k}a_{k+j_{2}-j_{1}}$.
\end{theorem}

Once we have obtained the asymptotic properties of the estimators of $a_{j}$%
, for $0\prec j$ and $j\leq M$, we are in a position to examine the
asymptotic properties of the predictor $\widehat{x}_{s}$ in $\left( \ref{1}%
\right) $ or $\left( \ref{1a}\right) $. To that end, denote by $\left\{
x_{t}^{\ast }\right\} _{t\in \mathbb{Z}^{2}}$ a new independent replicate
sequence with the same statistical properties of the original sequence $%
\left\{ x_{t}\right\} _{t\in \mathbb{Z}^{2}}$ not used in the estimation of
the spectral density function. Then let $\widehat{x}_{s}^{\ast }$ be as in $%
\left( \ref{1}\right) $ but with $\widehat{x}_{t}$ replaced by $x_{t}^{\ast
} $ there, that is%
\begin{equation*}
\widehat{x}_{s\left[ 1\right] ,s\left[ 2\right] }=-\sum_{k\left[ 2\right]
=1}^{M\left[ 2\right] }\widehat{a}_{0,k\left[ 2\right] }~x_{s\left[ 1\right]
,s\left[ 2\right] -k\left[ 2\right] }^{\ast }-\sum_{k\left[ 1\right] =1}^{M%
\left[ 1\right] }\sum_{k\left[ 2\right] =1-M\left[ 2\right] }^{M\left[ 2%
\right] }\widehat{a}_{k}x_{s-k}^{\ast }\text{,}
\end{equation*}%
or $\left( \ref{1a}\right) $ but with $\widehat{x}_{t}$ being replaced by $%
x_{t}^{\ast }$ there, that is 
\begin{equation*}
\widehat{x}_{t\left[ 1\right] -r\left[ 1\right] ,t\left[ 2\right] -r\left[ 2%
\right] }^{\ast }=-\sum_{k\left[ 2\right] =1}^{M\left[ 2\right] }\widehat{a}%
_{0,k\left[ 2\right] }~\widehat{x}_{t\left[ 1\right] -r\left[ 1\right] ,t%
\left[ 2\right] -r\left[ 2\right] -k\left[ 2\right] }^{\ast }-\sum_{k\left[ 1%
\right] =1}^{M\left[ 1\right] }\sum_{k\left[ 2\right] =1-M\left[ 2\right]
}^{M\left[ 2\right] }\widehat{a}_{k}x_{\left( s-r\right) -k}^{\ast }\text{.}
\end{equation*}

\begin{theorem}
Under $C1-C4$, we have that%
\begin{eqnarray*}
\left( \mathbf{a}\right) \text{ \ \ \ \ \ \ }AE\left( \widehat{x}_{s\left[ 1%
\right] ,s\left[ 2\right] }^{\ast }-x_{s\left[ 1\right] ,s\left[ 2\right]
}^{\ast }\right) ^{2} &=&\sigma _{\vartheta }^{2}\text{,} \\
\left( \mathbf{b}\right) \text{ }AE\left( \widehat{x}_{n\left[ 1\right] +1,t%
\left[ 2\right] }^{\ast }-x_{n\left[ 1\right] +1,t\left[ 2\right] }^{\ast
}\right) ^{2} &=&\left( 1+\sum_{k\left[ 2\right] =1}^{\infty }\zeta _{0,k%
\left[ 2\right] }^{2}\right) \sigma _{\vartheta }^{2}\text{,}
\end{eqnarray*}%
where $AE$ denotes the \textquotedblleft \emph{Asymptotic Expectation}%
\textquotedblright .
\end{theorem}

\section{Monte Carlo experiment}

\label{sec:MC} We examine the finite-sample behaviour of our algorithm in a
set of Monte Carlo simulations. As in \cite{Robinson2006} and \cite%
{robinson2007nonparametric} we used the model 
\begin{equation}
x_{t}=\epsilon _{t}+\tau \underset{s\neq 0}{\sum_{s_{1}=-1}^{1}%
\sum_{s_{2}=-1}^{1}}\epsilon _{t-s},  \label{mc_model}
\end{equation}%
similar to one considered in \cite{Haining1978}. Then 
\begin{equation}
f(\lambda )=\left( 2\pi \right) ^{-2}\left\{ 1+\tau \nu \left( \lambda
\right) \right\} ,  \label{mc_spec}
\end{equation}%
with $\nu \left( \lambda \right) =\prod_{j=1}^{2}\left( 1+2\cos \lambda
_{j}\right) -1$. \cite{Robinson2006} show that a sufficient condition for
invertibility of (\ref{mc_model}) is 
\begin{equation}
\left\vert \tau \right\vert <1/8.  \label{tauinv}
\end{equation}

We first generated a $40\times 41$ lattice using (\ref{mc_model}), with $%
\tau =0.05, 0.075, 0.10$ and the $\epsilon _{t}$ drawn independently from
three different distributions for each $\tau$: $U(-5,5)$, $N(0,1)$ and $\chi
_{9}^{2}-9$. The aim of this section is to examine the performance of both
prediction algorithms in predicting the $20,20$-th element of this lattice.
We did this by assuming a situation in which the practitioner has available
data sets of various sizes, generated from (\ref{mc_model}). To permit a
clear like-for-like comparison of improvement in performance as sample size
increases, we construct the prediction coefficients using the samples
generated in each replication and then use these to construct predictions
for the 20,20-th element of the $40\times 41$ lattice.

We took $n[1]=n^{\ast }+1$ and $n[2]=2n^{\ast }+1$, for some positive
integer $n^{\ast }$, implying $\mathbf{n}=\left( 2n^{\ast }+1\right)
(n^{\ast }+1)$, and generated iid $\epsilon_t$ from each of the three
distributions mentioned in the previous paragraph. In each of the 1000
replications we experimented with $\tau =0.05,0.075,0.10$ and $n^{\ast
}=5,10,20$ and $40$. The choices of $\tau $ satisfy (\ref{tauinv}).

Given the different sample sizes in each dimension, we can experiment with
more values of $m[1], m[2]$ and $p_1,p_2$ as $n^*$ increases. We make the
following choices: 
\begin{eqnarray*}
&&m[1]=m[2]=1; p^*=p_1=p_2=1,2, \text{ when } n^*=5, \\
&&m[1]=1,2;m[2]=1,2;p^*=p_1=p_2=1,2,3, \text{ when } n^*=10, \\
&&m[1]=1,2,3;m[2]=1,2,3,4,5;p^*=p_1=p_2=1,2,3, \text{ when } n^*=20, \\
&&m[1]=m[2]=1,2,3,4,5;p^*=p_1=p_2= 1,2,3,4,5 \text{ when } n^*=40.
\end{eqnarray*}

The flexible exponential approach requires a nonparametric estimate of $%
f(\lambda)$. Two such estimates are available to use: the first one based on
the tapered periodogram described in (\ref{fest}), which we denote $\hat{f}%
(\lambda)$, and the second based on the autoregressive approach in \cite%
{Gupta2016}. The latter also \label{sec:AR_pred} provides a rival prediction
methodology based on a nonparametric algorithm using AR model fitting,
extending well established results for $d=1$, see \cite{Bhansali1978} and 
\cite{Lewis1985}. The idea is first to obtain a least squares predictor
based on a truncated autoregression of order $p=\left(
p_{L_{1}},p_{U_{1}};p_{L_{2}},p_{U_{2}}\right) $, for non-negative integers $%
p_{L_{\ell }},\;p_{U_{\ell }}$, $\ell =1,2$, with the truncation allowed to
diverge as $N\rightarrow \infty $. That is, we approximate the infinite
unilateral representation in (\ref{arunil}) by one of increasing order.

In view of the half-plane representation we can \emph{a priori} set, say, $%
p_{L_{2}}=0$ when considering $\preccurlyeq $. If we could observe the AR
prediction coefficients $a_{k}$, say, a prediction of $x_{s}$ based on $%
\preccurlyeq $ could be constructed as 
\begin{equation}
\check{x}_{s}=\sum_{k\in S\left[ -p_{L},p_{U}\right] }a_{k}\check{x}_{s-k},
\label{AR_pred_infeasible_one}
\end{equation}%
where $S\left[ -p_{L},p_{U}\right]$ is the intersection of the set $\left\{
t\in \mathbb{L}:-p_{L_{\ell }}\leq t_{\ell }\ \leq p_{U_{\ell }},\;\ell
=1,2\right\}$ with the prediction half-plane. This is the spatial version of
one-step prediction and again we follow the convention that $\check{x}_s=x_s$
if $x_s$ is observed. However (\ref{AR_pred_infeasible_one}) is not feasible
and needs to be replaced by an approximate version, as described below.

Writing $p_{\ell }=p_{L_{\ell }}+p_{U_{\ell }}$, we assume throughout that $%
n[\ell ]>p_{\ell }$ for $\ell =1,2$, and denote $n_{p}=\prod_{\ell
=1}^{2}\left( n[\ell ]-p_{\ell }\right) $, $\mathfrak{h}(p)=p_{U_{2}}+\left(
p_{1}+1\right) p_{U_{2}}$ , i.e. the cardinality of $S\left[ -p_{L},p_{U}%
\right] $. Suppose that the data are observed on $\left\{ \left(
t_{1},t_{2}\right) :n_{L_{1}}\leq t_{1}\leq n_{U_{1}},-n_{L_{2}}\leq
t_{2}\leq n_{U_{2}}\right\} $. Define a least squares predictor of order $%
\mathfrak{h}(p)$ by 
\begin{equation}
\check{d}_{p}=\displaystyle{arg\min }_{a_{k},k\in S\left[ -p_{L},p_{U}\right]
}n_{p}^{-1}{\sum_{j(p,n)}{^{\prime \prime }}}\left( x_{j}-\sum_{k\in S\left[
-p_{L},p_{U}\right] }a_{k}x_{k-j}\right) ^{2},  \label{AR_LS_def}
\end{equation}%
where $\sum_{j(p,n)}^{\prime \prime }$ runs over $\left\{ \left(
j_{1},j_{2}\right) :p_{1}-n_{L_{1}}<j_{1}\leq
n_{U_{1}}+1,p_{2}-n_{L_{2}}<j_{2}\leq n_{U_{2}}+1\right\} $. We denote the
elements of $\check{d}_{p}$ by $\check{d}_{p}(k)$, $k\in S\left[ -p_{L},p_{U}%
\right] $, and the minimum value by $\check{\sigma}_{p}^{2}$. A feasible
half-plane prediction based on a fitted autoregression of order $p$ is given
by 
\begin{equation}
\check{x}_{p,s}=\sum_{k\in S\left[ -p_{L},p_{U}\right] }\check{d}_{p}(k)%
\check{x}_{s-k}.  \label{AR_pred_one}
\end{equation}%
The autoregressive nonparametric spectrum estimate is defined as 
\begin{equation}
\check{f}(\lambda )=\frac{\check{\sigma}_{p}^{2}}{(2\pi )^{2}\left\vert
1-\sum_{k\in S\left[ -p_{L},p_{U}\right] }\check{d}_{p}(k)e^{ik^{\prime
}\lambda }\right\vert ^{2}}.  \label{fest_AR}
\end{equation}%
A predictor of $x_{s}$ based on $\left( \ref{1}\right) $ using $\hat{f}%
(\lambda )$ (respectively $\check{f}(\lambda )$) is denoted $\hat{x}_{s}$
(respectively $\tilde{x}_{s}$), while a predictor based on (\ref{AR_pred_one}%
) is denoted $\check{x}_{s}$ as mentioned above.

Let $\vec{x}_{r,s}$ be a generic predictor of $x_{s}$ in replication $r$, $%
r=1,\ldots ,1000$. We report a statistic called the root mean squared error
(RMSE) of prediction, defined as 
\begin{equation}
\text{\emph{RMSE}}(\vec{x}_s)=\left\{ \frac{1}{1000}\sum_{r=1}^{1000}\left( 
\vec{x}_{r,s}-x_{s}\right) ^{2}\right\} ^{\frac{1}{2}}.  \label{RMSE_def}
\end{equation}
The results are reported in Tables \ref{table_RMSE_n5_C}-\ref%
{table_RMSE_n40_C}. We observe an improvement in prediction performance as $%
n^*$ increases, and also as the bandwidths ($(m[1],m[2])$ and $p^*$)
increase as function of $n^*$. This is as expected in the theory.
Nevertheless, even for rather small sample sizes the RMSE is acceptable. For
example, for $\epsilon_t\sim U(-5,5)$ and $\epsilon_t\sim N(0,1)$ with $%
n^*=5 $, we can obtain predictions with RMSE that are not radically
different from the $n^*=10$ case, even though this change in $n^*$ entails a
sample that is nearly four times larger (231 against 66). In comparison the
RMSE with the smaller sample size can be quite close to those obtained with
more data in some cases, cf. $\check x _{20,20}$ for any error distribution.

For the smallest sample size $\check x_{20,20}$ can outperform $\hat x_{
20,20}$ and $\tilde x _{20,20}$, but with increasing $n^*$ the latter two
clearly begin to dominate. An inspection of Tables \ref{table_RMSE_n5_C}-\ref%
{table_RMSE_n40_C} reveals that the use of the flexible exponential
algorithm proposed in this paper together with either the tapered
periodogram or the AR spectral estimator of \cite{Gupta2016} outperforms
autoregressive prediction in moderate to large sample sizes. There is little
to choose from between the two best performing algorithms, and a
practitioner might choose to use either one. However the AR prediction is
clearly dominated by our algorithm.

\begin{table}[tbp]
{\small {\ }}
\par
\begin{center}
{\small 
\begin{tabular}{ccccccccccc}
$\epsilon_t\sim U(-5,5)$ &  &  &  &  &  &  &  &  &  &  \\ 
$\tau$ &  & {0.05} & {0.075} & {0.10} & {0.05} & {0.075} & {0.10} & {0.05} & 
{0.075} & {0.10} \\ 
\midrule $\left(m[1],m[2]\right)$ & {$p^*$} & \multicolumn{3}{c}{$\hat
x_{20,20}$} & \multicolumn{3}{c}{$\tilde x_{20,20}$} & \multicolumn{3}{c}{$%
\check x_{20,20}$} \\ 
\midrule (1,1) & 1 & 0.5273 & 0.5252 & 0.5269 & 0.5143 & 0.5123 & 0.5144 & 
0.4279 & 0.4575 & 0.4813 \\ 
(1,1) & 2 &  &  &  & 0.5141 & 0.5123 & 0.5144 & 0.4806 & 0.5051 & 0.5261 \\ 
\midrule $\epsilon_t\sim N(0,1)$ &  &  &  &  &  &  &  &  &  &  \\ 
$\tau$ &  & {0.05} & {0.075} & {0.10} & {0.05} & {0.075} & {0.10} & {0.05} & 
{0.075} & {0.10} \\ 
\midrule $\left(m[1],m[2]\right)$ & {$p^*$} & \multicolumn{3}{c}{$\hat
x_{20,20}$} & \multicolumn{3}{c}{$\tilde x_{20,20}$} & \multicolumn{3}{c}{$%
\check x_{20,20}$} \\ 
\midrule (1,1) & 1 & 1.2916 & 1.1360 & 1.1072 & 1.2713 & 1.1120 & 1.0810 & 
1.0487 & 1.0226 & 0.9874 \\ 
(1,1) & 2 &  &  &  & 1.2712 & 1.1120 & 1.0811 & 1.0829 & 1.0589 & 1.0313 \\ 
\midrule $\epsilon_t\sim \chi^2_9-9$ &  &  &  &  &  &  &  &  &  &  \\ 
$\tau$ &  & {0.05} & {0.075} & {0.10} & {0.05} & {0.075} & {0.10} & {0.05} & 
{0.075} & {0.10} \\ 
\midrule $\left(m[1],m[2]\right)$ & {$p^*$} & \multicolumn{3}{c}{$\hat
x_{20,20}$} & \multicolumn{3}{c}{$\tilde x_{20,20}$} & \multicolumn{3}{c}{$%
\check x_{20,20}$} \\ 
\midrule (1,1) & 1 & 2.0651 & 1.9869 & 1.8656 & 2.0199 & 1.9435 & 1.8238 & 
2.2475 & 2.1666 & 2.0835 \\ 
(1,1) & 2 &  &  &  & 2.0197 & 1.9435 & 1.8239 & 2.5720 & 2.5355 & 2.4353 \\ 
\bottomrule &  &  &  &  &  &  &  &  &  & 
\end{tabular}
\vspace{0.0in} }
\end{center}
\par
{\small \ }
\caption{Monte Carlo RMSE of prediction with $n^*=5$, model (\protect\ref%
{mc_model})}
\label{table_RMSE_n5_C}
\end{table}

\begin{table}[tbp]
{\small {\ }}
\par
\begin{center}
{\small 
\begin{tabular}{ccccccccccc}
$\epsilon_t\sim U(-5,5)$ &  &  &  &  &  &  &  &  &  &  \\ 
$\tau$ &  & {0.05} & {0.075} & {0.10} & {0.05} & {0.075} & {0.10} & {0.05} & 
{0.075} & {0.10} \\ 
\midrule $\left(m[1],m[2]\right)$ & {$p*$} & \multicolumn{3}{c}{$\hat
x_{20,20}$} & \multicolumn{3}{c}{$\tilde x_{20,20}$} & \multicolumn{3}{c}{$%
\check x_{20,20}$} \\ 
\midrule (1,1) & 1 & 0.4212 & 0.4200 & 0.4158 & 0.3999 & 0.3989 & 0.3946 & 
0.4250 & 0.4554 & 0.4808 \\ 
(2,2) & 2 & 0.5288 & 0.5308 & 0.5336 & 0.5161 & 0.5184 & 0.5217 & 0.4325 & 
0.4626 & 0.4896 \\ 
(1,2) & 3 & 0.3859 & 0.3806 & 0.3776 & 0.3849 & 0.3792 & 0.3757 & 0.4390 & 
0.4630 & 0.4894 \\ 
\midrule $\epsilon_t\sim N(0,1)$ &  &  &  &  &  &  &  &  &  &  \\ 
$\tau$ &  & {0.05} & {0.075} & {0.10} & {0.05} & {0.075} & {0.10} & {0.05} & 
{0.075} & {0.10} \\ 
\midrule $\left(m[1],m[2]\right)$ & {$p^*$} & \multicolumn{3}{c}{$\hat
x_{20,20}$} & \multicolumn{3}{c}{$\tilde x_{20,20}$} & \multicolumn{3}{c}{$%
\check x_{20,20}$} \\ 
\midrule (1,1) & 1 & 1.2711 & 1.3066 & 1.2479 & 1.2487 & 1.2862 & 1.2283 & 
1.0517 & 1.0152 & 0.9788 \\ 
(2,2) & 2 & 1.1526 & 1.1248 & 1.0953 & 1.1318 & 1.1016 & 1.0700 & 1.0575 & 
1.0252 & 0.9912 \\ 
(1,2) & 3 & 1.2689 & 1.2123 & 1.1755 & 1.2433 & 1.1868 & 1.1501 & 1.0700 & 
1.0376 & 1.0006 \\ 
\midrule $\epsilon_t\sim \chi^2_9-9$ &  &  &  &  &  &  &  &  &  &  \\ 
$\tau$ &  & {0.05} & {0.075} & {0.10} & {0.05} & {0.075} & {0.10} & {0.05} & 
{0.075} & {0.10} \\ 
\midrule $\left(m[1],m[2]\right)$ & {$p^*$} & \multicolumn{3}{c}{$\hat
x_{20,20}$} & \multicolumn{3}{c}{$\tilde x_{20,20}$} & \multicolumn{3}{c}{$%
\check x_{20,20}$} \\ 
\midrule (1,1) & 1 & 0.9720 & 0.8856 & 0.8149 & 0.9805 & 0.8932 & 0.8217 & 
2.0183 & 1.9234 & 1.8353 \\ 
(2,2) & 2 & 2.0694 & 1.9571 & 1.8487 & 2.0256 & 1.9150 & 1.8080 & 2.0018 & 
1.8313 & 1.6798 \\ 
(1,2) & 3 & 1.3581 & 1.2594 & 1.1650 & 1.5024 & 1.4011 & 1.3045 & 2.1316 & 
1.9801 & 1.8485 \\ 
\bottomrule &  &  &  &  &  &  &  &  &  & 
\end{tabular}
}
\end{center}
\par
{\small \ }
\caption{Monte Carlo RMSE of prediction with $n^*=10$, model (\protect\ref%
{mc_model})}
\label{table_RMSE_n10_C}
\end{table}

\begin{table}[tbp]
{\small {\ }}
\par
\begin{center}
{\small 
\begin{tabular}{ccccccccccc}
$\epsilon_t\sim U(-1,1)$ &  &  &  &  &  &  &  &  &  &  \\ 
$\tau$ &  & {0.05} & {0.075} & {0.10} & {0.05} & {0.075} & {0.10} & {0.05} & 
{0.075} & {0.10} \\ 
\midrule $\left(m[1],m[2]\right)$ & {$p^*$} & \multicolumn{3}{c}{$\hat
x_{20,20}$} & \multicolumn{3}{c}{$\tilde x_{20,20}$} & \multicolumn{3}{c}{$%
\check x_{20,20}$} \\ 
\midrule (1,1) & 1 & 0.2946 & 0.2865 & 0.2806 & 0.2989 & 0.2905 & 0.2843 & 
0.4245 & 0.4549 & 0.4806 \\ 
(1,2) & 2 & 0.3917 & 0.3861 & 0.3817 & 0.3942 & 0.3881 & 0.3832 & 0.4263 & 
0.4587 & 0.4867 \\ 
(1,3) & 2 & 0.4526 & 0.4427 & 0.4335 & 0.4448 & 0.4350 & 0.4258 &  &  &  \\ 
(2,3) & 2 & 0.4273 & 0.4170 & 0.4084 & 0.4136 & 0.4036 & 0.3953 &  &  &  \\ 
(2,4) & 2 & 0.3792 & 0.3763 & 0.3732 & 0.3783 & 0.3749 & 0.3713 &  &  &  \\ 
(3,4) & 4 & 0.4274 & 0.4238 & 0.4216 & 0.4246 & 0.4203 & 0.4174 & 0.4326 & 
0.4617 & 0.4889 \\ 
(3,5) & 3 & 0.4270 & 0.4238 & 0.4216 & 0.4243 & 0.4203 & 0.4173 & 0.4271 & 
0.4593 & 0.4871 \\ 
\midrule $\epsilon_t\sim N(0,1)$ &  &  &  &  &  &  &  &  &  &  \\ 
$\tau$ &  & {0.05} & {0.075} & {0.10} & {0.05} & {0.075} & {0.10} & {0.05} & 
{0.075} & {0.10} \\ 
\midrule $\left(m[1],m[2]\right)$ & {$p^*$} & \multicolumn{3}{c}{$\hat
x_{20,20}$} & \multicolumn{3}{c}{$\tilde x_{20,20}$} & \multicolumn{3}{c}{$%
\check x_{20,20}$} \\ 
\midrule (1,1) & 1 & 1.1985 & 1.2252 & 1.1586 & 1.2006 & 1.2274 & 1.1612 & 
1.0515 & 1.0146 & 0.9783 \\ 
(1,2) & 2 & 1.2035 & 0.9851 & 0.9384 & 1.1964 & 0.9775 & 0.9311 & 1.0532 & 
1.0145 & 0.9800 \\ 
(1,3) & 2 & 1.0145 & 0.8611 & 0.8164 & 1.0234 & 0.8709 & 0.8262 &  &  &  \\ 
(2,3) & 2 & 0.9794 & 0.9397 & 0.8998 & 0.9850 & 0.9457 & 0.9061 &  &  &  \\ 
(2,4) & 2 & 1.2033 & 1.1699 & 1.1390 & 1.1771 & 1.1439 & 1.1132 &  &  &  \\ 
(3,4) & 4 & 1.1636 & 1.1374 & 1.1112 & 1.1456 & 1.1187 & 1.0919 & 1.0572 & 
1.0161 & 0.9786 \\ 
(3,5) & 3 & 1.1634 & 1.1371 & 1.1110 & 1.1454 & 1.1185 & 1.0917 & 1.0594 & 
1.0170 & 0.9795 \\ 
\midrule $\epsilon_t\sim \chi^2_9-9$ &  &  &  &  &  &  &  &  &  &  \\ 
$\tau$ &  & {0.05} & {0.075} & {0.10} & {0.05} & {0.075} & {0.10} & {0.05} & 
{0.075} & {0.10} \\ 
\midrule $\left(m[1],m[2]\right)$ & {$p^*$} & \multicolumn{3}{c}{$\hat
x_{20,20}$} & \multicolumn{3}{c}{$\tilde x_{20,20}$} & \multicolumn{3}{c}{$%
\check x_{20,20}$} \\ 
\midrule (1,1) & 1 & 1.0012 & 0.5964 & 0.5140 & 1.0421 & 0.6479 & 0.5613 & 
1.9713 & 1.8688 & 1.7806 \\ 
(1,2) & 2 & 0.6229 & 0.5203 & 0.4430 & 0.7138 & 0.6053 & 0.5215 & 1.9134 & 
1.7618 & 1.5989 \\ 
(1,3) & 2 & 0.9546 & 0.8202 & 0.6878 & 1.0768 & 0.9393 & 0.8037 &  &  &  \\ 
(2,3) & 2 & 0.9984 & 0.8924 & 0.7884 & 1.1088 & 0.9986 & 0.8906 &  &  &  \\ 
(2,4) & 2 & 1.3377 & 1.2419 & 1.1488 & 1.4835 & 1.3859 & 1.2906 &  &  &  \\ 
(3,4) & 4 & 1.7939 & 1.6942 & 1.5952 & 1.7852 & 1.6818 & 1.5790 & 1.9971 & 
1.8303 & 1.6848 \\ 
(3,5) & 3 & 1.7913 & 1.6904 & 1.5911 & 1.7825 & 1.6780 & 1.5749 & 1.9484 & 
1.7935 & 1.6429 \\ 
\bottomrule &  &  &  &  &  &  &  &  &  & 
\end{tabular}
\vspace{0.0in} }
\end{center}
\par
{\small \ }
\caption{Monte Carlo RMSE of prediction with $n^*=20$, model (\protect\ref%
{mc_model})}
\label{table_RMSE_n20_C}
\end{table}

\begin{table}[t]
{\small {\ }}
\par
\begin{center}
{\small 
\begin{tabular}{ccccccccccc}
$\epsilon_t\sim U(-5,5)$ &  &  &  &  &  &  &  &  &  &  \\ 
$\tau$ &  & {0.05} & {0.075} & {0.10} & {0.05} & {0.075} & {0.10} & {0.05} & 
{0.075} & {0.10} \\ 
\midrule $\left(m[1],m[2]\right)$ & {$p^*$} & \multicolumn{3}{c}{$\hat
x_{20,20}$} & \multicolumn{3}{c}{$\tilde x_{20,20}$} & \multicolumn{3}{c}{$%
\check x_{20,20}$} \\ 
\midrule (1,1) & 1 & 0.2034 & 0.2049 & 0.2066 & 0.2138 & 0.2151 & 0.2165 & 
0.4228 & 0.4541 & 0.4800 \\ 
(2,2) & 2 & 0.2880 & 0.2821 & 0.2763 & 0.2924 & 0.2862 & 0.2801 & 0.4234 & 
0.4561 & 0.4844 \\ 
(3,3) & 3 & 0.3229 & 0.3133 & 0.3039 & 0.3100 & 0.3008 & 0.2917 & 0.4234 & 
0.4560 & 0.4843 \\ 
(4,4) & 4 & 0.3987 & 0.3946 & 0.3907 & 0.3768 & 0.3728 & 0.3691 & 0.4243 & 
0.4567 & 0.4848 \\ 
(5,5) & 5 & 0.3766 & 0.3712 & 0.3656 & 0.3734 & 0.3678 & 0.3619 & 0.4253 & 
0.4575 & 0.4852 \\ 
\midrule $\epsilon_t\sim N(0,1)$ &  &  &  &  &  &  &  &  &  &  \\ 
$\tau$ &  & {0.05} & {0.075} & {0.10} & {0.05} & {0.075} & {0.10} & {0.05} & 
{0.075} & {0.10} \\ 
\midrule $\left(m[1],m[2]\right)$ & {$p^*$} & \multicolumn{3}{c}{$\hat
x_{20,20}$} & \multicolumn{3}{c}{$\tilde x_{20,20}$} & \multicolumn{3}{c}{$%
\check x_{20,20}$} \\ 
\midrule (1,1) & 1 & 1.0378 & 1.0174 & 0.9847 & 1.0430 & 1.0226 & 0.9897 & 
1.0516 & 1.0145 & 0.9782 \\ 
(2,2) & 2 & 1.0132 & 0.9739 & 0.9349 & 1.0156 & 0.9765 & 0.9379 & 1.0512 & 
1.0155 & 0.9809 \\ 
(3,3) & 3 & 1.0971 & 1.0582 & 1.0201 & 1.0716 & 1.0331 & 0.9954 & 1.0508 & 
1.0134 & 0.9755 \\ 
(4,4) & 4 & 1.1024 & 1.0610 & 1.0191 & 1.0782 & 1.0377 & 0.9967 & 1.0511 & 
1.0136 & 0.9758 \\ 
(5,5) & 5 & 1.0987 & 1.0490 & 0.9989 & 1.0538 & 1.0054 & 0.9566 & 1.0517 & 
1.0141 & 0.9760 \\ 
\midrule $\epsilon_t\sim \chi^2_9-9$ &  &  &  &  &  &  &  &  &  &  \\ 
$\tau$ &  & {0.05} & {0.075} & {0.10} & {0.05} & {0.075} & {0.10} & {0.05} & 
{0.075} & {0.10} \\ 
\midrule $\left(m[1],m[2]\right)$ & {$p^*$} & \multicolumn{3}{c}{$\hat
x_{20,20}$} & \multicolumn{3}{c}{$\tilde x_{20,20}$} & \multicolumn{3}{c}{$%
\check x_{20,20}$} \\ 
\midrule (1,1) & 1 & 1.0228 & 0.9107 & 0.7998 & 1.0813 & 0.9668 & 0.8535 & 
1.9533 & 1.8618 & 1.7730 \\ 
(2,2) & 2 & 0.5742 & 0.4864 & 0.4025 & 0.6323 & 0.5412 & 0.4533 & 1.8933 & 
1.7385 & 1.5718 \\ 
(3,3) & 3 & 0.7560 & 0.6947 & 0.6335 & 0.8305 & 0.7655 & 0.7007 & 1.9050 & 
1.7582 & 1.6037 \\ 
(4,4) & 4 & 0.8226 & 0.7490 & 0.6752 & 0.8315 & 0.7571 & 0.6825 & 1.9156 & 
1.7701 & 1.6184 \\ 
(5,5) & 5 & 0.8275 & 0.7496 & 0.6720 & 0.9196 & 0.8378 & 0.7563 & 1.9327 & 
1.7883 & 1.6380 \\ 
\bottomrule &  &  &  &  &  &  &  &  &  & 
\end{tabular}
\vspace{0.0in} }
\end{center}
\par
{\small \ }
\caption{Monte Carlo RMSE of prediction with $n^*=40$, model (\protect\ref%
{mc_model})}
\label{table_RMSE_n40_C}
\end{table}

\section{An application to house price prediction in Los Angeles}

\label{sec:empirical} In this section we show how the techniques established
in the paper can be used to predict house prices. This can be of interest in
real estate and urban economics, as well as for property developers. Indeed,
spatial methods are frequently used in these fields, as studied for instance
by \cite{IversenJr2001}, \cite{Banerjee2004} and \cite{Majumdar2006}. We use
median house price data for census blocks in California from the 1990 census
from \cite{Pace1997a}, available at www.spatial-statistics.com. We confine
our analysis to the city of Los Angeles. The data is gridded as follows: a $%
14\times 23$ grid of square cells is superimposed on Los Angeles, from $%
33.75^{\circ}$N to $34.17^{\circ}$N and $117.75^{\circ}$W to $118.44^{\circ}$%
W. The grid covers a total of 5259 observations. The average of the median
house values for each cell is calculated and the 322 such observations form
our sample. The gridding is shown in Figure \ref{fig:la_data}, in which the
8 empty cells are filled and marked with a cross. We wish to predict the
house price for these cells. House price data is not a zero mean process, so
we subtract the sample mean using the whole sample from each cell.

We proceed in the following way: to obtain the coefficients $\hat{a}_{\ell }$%
, and $\check{d}(\ell )$ in $\left( \ref{1}\right) $ and (\ref{AR_pred_one})
we use the $14\times 19$ sublattice formed of the first 19 columns of cells.
This sublattice contains no missing observations. Once the coefficients are
obtained we construct predictions using the remaining $4\times 19$
sublattice, in a step-by-step manner. The shaded-and-crossed cell (8,20) is
predicted first, followed by (8,21) and (8,22). We then predict (4,21),
followed by (7,22), (9,23), (6,23) and (1,23).

The predicted values are tabulated for various values of $\left(
m[1],m[2]\right) $ and $\left( p_{1},p_{2}\right) $ in Tables \ref%
{table_LA_preds_Javier} and \ref{table_LA_preds_AR}. The predicted values
are quite stable across the choices $\left( m[1],m[2]\right) =(2,2),(2,3)$
using either the periodogram or AR spectral estimate. They most closely
match those obtained when $\left( p_{1},p_{2}\right) =(2,2),(2,3)$ in (\ref%
{AR_pred_one}). In the latter case we compare in Table \ref%
{table_LA_BIC_FPE_AR} the order selection criteria proposed by \cite%
{Gupta2016}, which include the usual FPE and BIC (denoted with a $\widehat{}$
) as well as corrected version that account for the spatial case (denoted
with $\widetilde{}$ and $\bar{}$ ). The FPE tends to favour longer lag
lengths no matter which version is used, as do $\widehat{\text{BIC}}$ and $%
\overline{\text{BIC}}$. However the latter as well as $\widetilde{\text{FPE}}
$ are not monotonically decreasing in lag length, unlike $\widehat{BIC}$, $%
\widehat{FPE}$ and $\overline{FPE}$. Thus the latter three are likely to
overfit and seem undesirable. If we impose a selection rule that picks the
desirable lag order as the first instance when the selection criteria shows
an increase with lag length, then we get $\left( p_{1},p_{2}\right) =(2,2)$
using $\widetilde{FPE}$ and $\overline{BIC} $. $\widetilde{BIC}$ indicates a
choice of $\left( p_{1},p_{2}\right) =(1,2)$, on the other hand. All
considered, it seems that $\left( p_{1},p_{2}\right) =(2,2)$ is a reasonable
choice. 
\begin{figure}[b]
\includegraphics[scale=0.6, trim=-180 0 -200 0,
clip=true]{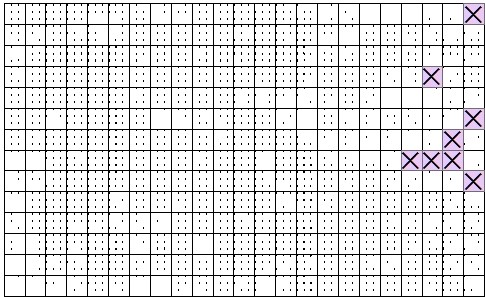}
\caption{Gridded Los Angeles median house price data}
\label{fig:la_data}
\end{figure}
\begin{table}[b]
{\small {\ }}
\par
\begin{center}
{\small 
\begin{tabular}{ccccccccccc}
\midrule $\left(m[1],m[2]\right)$ &  & (8,20) & (8,21) & (8,22) & (4,21) & 
(7,22) & (9,23) & (6,23) & (1,23) &  \\ 
\midrule (1,1) & $\widehat{}$ & 0.5645 & 0.5540 & 0.5384 & 0.6976 & 0.6671 & 
0.1962 & 0.9337 & 0.4518 &  \\ 
& $\widetilde{}$ & 0.5244 & 0.5529 & 0.5523 & 0.6643 & 0.6256 & 0.1969 & 
0.8770 & 0.4348 &  \\ 
(1,2) & $\widehat{}$ & 0.8054 & 0.5312 & 0.8791 & 1.2740 & 0.8797 & 1.0049 & 
1.2029 & 0.6338 &  \\ 
& $\widetilde{}$ & 0.7043 & 0.4290 & 0.7684 & 1.1141 & 0.7693 & 0.8788 & 
1.0520 & 0.5543 &  \\ 
(2,1) & $\widehat{}$ & 1.0691 & 0.8761 & 0.3926 & 1.4462 & 1.9503 & 0.7717 & 
0.8519 & 1.0789 &  \\ 
& $\widetilde{}$ & 0.9828 & 0.8133 & 0.4722 & 1.3612 & 1.7429 & 0.6726 & 
0.8448 & 1.0044 &  \\ 
(2,2) & $\widehat{}$ & 2.0717 & 1.7439 & 1.4680 & 3.0421 & 2.4033 & 2.3360 & 
2.3242 & 2.0935 &  \\ 
& $\widetilde{}$ & 1.8934 & 1.4566 & 1.1207 & 2.7803 & 2.1965 & 2.1349 & 
2.1242 & 1.9134 &  \\ 
(2,3) & $\widehat{}$ & 1.9970 & 1.6205 & 1.3149 & 2.9325 & 2.3167 & 2.2518 & 
2.2404 & 2.0181 &  \\ 
& $\widetilde{}$ & 1.7926 & 1.3057 & 0.9511 & 2.6323 & 2.0796 & 2.0213 & 
2.0111 & 1.8115 &  \\ 
(1,3) & $\widehat{}$ & 0.8150 & 0.5414 & 0.8899 & 1.2892 & 0.8902 & 1.0169 & 
1.2173 & 0.6414 &  \\ 
& $\widetilde{}$ & 0.7218 & 0.4460 & 0.7874 & 1.1419 & 0.7885 & 0.9007 & 
1.0782 & 0.5681 &  \\ 
\bottomrule &  &  &  &  &  &  &  &  &  & 
\end{tabular}
\vspace{0.0in} }
\end{center}
\par
{\small \ }
\caption{Los Angeles house price predictions, in
`00,000 US Dollars}
\label{table_LA_preds_Javier}
\end{table}

\begin{table}[t]
{\small {\ }}
\par
\begin{center}
{\small 
\begin{tabular}{ccccccccccc}
\midrule $\left(p_1,p_2\right)$ & (8,20) & (8,21) & (8,22) & (4,21) & (7,22)
& (9,23) & (6,23) & (1,23) &  &  \\ 
\midrule (1,1) & 1.6525 & 1.2876 & 1.0325 & 2.4728 & 1.6663 & 2.1204 & 1.6825
& 1.7266 &  &  \\ 
(1,2) & 1.5373 & 1.2516 & 0.8473 & 2.3741 & 1.7944 & 1.8452 & 1.7626 & 1.5659
&  &  \\ 
(2,1) & 1.8138 & 1.6433 & 1.1237 & 2.4881 & 1.8545 & 2.3404 & 1.7547 & 2.3233
&  &  \\ 
(2,2) & 1.8703 & 1.9858 & 1.4089 & 2.6079 & 2.4100 & 2.4992 & 2.0862 & 2.2252
&  &  \\ 
(3,2) & 1.7925 & 1.9298 & 1.4400 & 2.4489 & 2.6472 & 2.5962 & 2.4021 & 1.9518
&  &  \\ 
(4,3) & 2.4465 & 2.2325 & 1.9075 & 2.0319 & 3.1207 & 3.5439 & 2.8234 & 2.0841
&  &  \\ 
\bottomrule &  &  &  &  &  &  &  &  &  & 
\end{tabular}
\vspace{0.0in} }
\end{center}
\par
{\small \ }
\caption{Los Angeles house price predictions using (\protect\ref{AR_pred_one}%
), in `00,000 US Dollars}
\label{table_LA_preds_AR}
\end{table}

\begin{table}[t]
{\small {\ }}
\par
\begin{center}
{\small 
\begin{tabular}{ccccccccccc}
\midrule $\left(p_1,p_2\right)$ & $\widehat{\text{BIC}}$ & $\widetilde{%
\text {BIC}}$ & $\overline{\text{BIC}}$ & $\widehat{\text{FPE}}$ & $%
\widetilde{\text{FPE}}$ & $\overline{\text{FPE}}$ &  &  &  &  \\ 
\midrule (1,1) & 0.5979 & 0.6001 & 0.5990 & 0.5472 & 0.5639 & 0.5555 &  &  & 
&  \\ 
(1,2) & 0.5973 & 0.5979 & 0.5976 & 0.5144 & 0.5183 & 0.5164 &  &  &  &  \\ 
(2,1) & 0.5953 & 0.6042 & 0.5997 & 0.4766 & 0.5377 & 0.5062 &  &  &  &  \\ 
(2,2) & 0.5946 & 0.6008 & 0.5977 & 0.4351 & 0.4728 & 0.4535 &  &  &  &  \\ 
(3,2) & 0.5943 & 0.6105 & 0.6024 & 0.4022 & 0.5018 & 0.4489 &  &  &  &  \\ 
(4,3) & 0.5824 & 0.6081 & 0.5953 & 0.2129 & 0.3058 & 0.2543 &  &  &  &  \\ 
\bottomrule &  &  &  &  &  &  &  &  &  & 
\end{tabular}
\vspace{0.0in} }
\end{center}
\par
{\small \ }
\caption{Los Angeles house price predictions using (\protect\ref{AR_pred_one}%
), BIC and FPE}
\label{table_LA_BIC_FPE_AR}
\end{table}

\section{Conclusion}

\label{sec:conc}

In this paper we have dealt with the problem of prediction when the data $%
\left\{ x_{t}\right\} _{t\in \mathbb{Z}^{2}}$ is collected on a lattice. To
do so, we considered unilateral representations of $\left\{ x_{t}\right\}
_{t\in \mathbb{Z}^{2}}$ and in particular the canonical factorization of the
spectral density function, the latter being possible as observed by \cite%
{Whittle1954}. Our approach does not need any parameterization of the model
(i.e. the covariogram structure of the data), so we avoid the consequences
that a wrong parameterization can have in the predictor. We have also
compared our methodology to one based on the space domain by using a finite
approximation of the unilateral autoregressive model in $\left( \ref{SAR}%
\right) $.

However, it might be interesting to examine how our proposed methodology
compares with one based on the conditional autoregressive ($CAR$)
representation of \cite{Besag1974}. That is, let $x_{t}$ be given by%
\begin{eqnarray}
x_{t} &=&\mu +E\left[ x_{t}\mid x_{r};r\neq t\right] +u_{t}  \notag \\
&=&\mu +\sum_{r\neq t}\zeta _{\left\vert r-t\right\vert }x_{r}+u_{t}\text{.}
\label{CAR}
\end{eqnarray}%
Note that our definition in $\left( \ref{CAR}\right) $ implies that $x_{t}$
is, among other characteristics, homogeneous. The representation of $x_{t}$
given in $\left( \ref{CAR}\right) $ suggests to predict a value $x_{t}$ at a
location $s=\left( s\left[ 1\right] ,s\left[ 2\right] \right) $, $1\leq s%
\left[ 1\right] \leq n\left[ 1\right] $ and $1\leq s\left[ 2\right] \leq n%
\left[ 2\right] $, by 
\begin{equation}
\widehat{x}_{s}=\widehat{\mu }+\sum_{r\neq s;\left\vert r-s\right\vert <M}%
\widehat{\zeta }_{\left\vert r-s\right\vert }x_{r}\text{,}  \label{1_1}
\end{equation}%
where $\widehat{\mu }$ and $\widehat{\zeta }_{\left\vert r-s\right\vert }$
are respectively the least squares estimator of $\mu $ and $\zeta
_{\left\vert r-t\right\vert }$, and with the convention that $x_{r}=0$ if it
were not observed. This is in the same spirit as we did with our predictor
in $\left( \ref{1}\right) $. On the other hand, if we were interesting to
predict a value $x_{t}$ at a location $s=\left( n\left[ 1\right] +1,s\left[ 2%
\right] \right) $, we might then use%
\begin{equation}
\widehat{x}_{n\left[ 1\right] +1,s\left[ 2\right] }=\widehat{\mu }%
+\sum_{r\neq s;\left\vert r-s\right\vert <M}\widehat{\zeta }_{\left\vert
r-s\right\vert }x_{r}\text{.}  \label{1a_1}
\end{equation}%
However to compute the prediction we would also need to replace the
unobserved $x_{r}$ by its prediction as in $\left( \ref{1a}%
\right) $. The latter might be done in an iterative fashion similar to what
we did in $\left( \ref{pre_1}\right) $. 

. \clearpage\appendix

\begin{center}
{\Large {\textbf{Mathematical Appendix}} }
\end{center}

\section{Proofs of Theorems}

For the sake of notational simplicity, we shall assume that $M\left[ 1\right]
=M\left[ 2\right] $ and also that $n\left[ 1\right] =n\left[ 2\right] $, so
that $\mathbf{n}=n^{2}\left[ 1\right] $ and $\mathbf{m}=m^{2}\left[ 1\right] 
$ say. Also to simplify the notation we shall write $\sum_{j\preceq J}$
instead of $\sum_{j\preceq J}^{+}$ given in $\left( \ref{notd}\right) $.
That is, 
\begin{equation}
\sum_{k\preceq M}d_{k}=\sum_{k\left[ 2\right] =1}^{M\left[ 2\right] }d_{0,k%
\left[ 2\right] }+\sum_{k\left[ 1\right] =1}^{M\left[ 1\right] }\sum_{k\left[
2\right] =1-M\left[ 2\right] }^{M\left[ 2\right] }d_{k\left[ 1\right] ,k%
\left[ 2\right] }\text{.}  \label{sum_1}
\end{equation}

\subsection{\textbf{Proof of Theorem 1}}

$\left. {}\right. $

We shall examine part $\left( \mathbf{a}\right) $, since part $\left( 
\mathbf{b}\right) $ follows by Lemma \ref{c_hat} and standard arguments. By
the Cram\'{e}r-Wold device, it suffices to show that for a finite set of
constants $\varphi _{j}$, $j=1,...,J$, 
\begin{equation}
\mathbf{n}^{1/2}\sum_{j=1}^{J}\varphi _{j}\left( \widehat{\alpha }_{j}-%
\widetilde{\alpha }_{j,n}\right) \overset{d}{\rightarrow }\mathcal{N}\left(
0,\sum_{j=1}^{J}\varphi _{j}^{2}\left( 1+\left( 1+\kappa _{4,\vartheta
}\right) \delta _{j}\right) \right) \text{.}  \label{pTh1.1}
\end{equation}%
First, by definition of $\widehat{\alpha }_{j}$ and $\widetilde{\alpha }%
_{j,n}$, we have that 
\begin{equation}
\widehat{\alpha }_{j}-\widetilde{\alpha }_{j,n}=\frac{1}{2\mathbf{M}}%
\sum_{k\preceq M}\log \left( \frac{\widehat{f}_{k}}{\widetilde{f}_{k}}%
\right) \cos \left( j\cdot \widetilde{\lambda }_{k}\right) \text{.}
\label{pTh1.2}
\end{equation}%
Because standard inequalities and then Lemma \ref{f_hat} yield that 
\begin{equation}
\sup_{k\preceq M}\left\vert \frac{\widehat{f}_{k}-\widetilde{f}_{k}}{%
\widetilde{f}_{k}}\right\vert ^{2}\leq \sum_{k\preceq M}\left\vert \frac{%
\widehat{f}_{k}-\widetilde{f}_{k}}{\widetilde{f}_{k}}\right\vert
^{2}=O_{p}\left( \frac{\mathbf{M}}{\mathbf{m}}\right) =o_{p}\left( 1\right) 
\text{,}  \label{pTh1.3}
\end{equation}%
the left side of $\left( \ref{pTh1.2}\right) $ is, by Lemma \ref{f_hat},%
\begin{eqnarray}
&&\frac{1}{2\mathbf{M}}\sum_{k\preceq M}\frac{\widehat{f}_{k}-\widetilde{f}%
_{k}}{\widetilde{f}_{k}}\cos \left( j\cdot \widetilde{\lambda }_{k}\right)
+O_{p}\left( \mathbf{m}^{-1}\right)  \notag \\
&=&\frac{1}{2\mathbf{M}}\sum_{k\preceq M}\frac{\widehat{f}_{k}-\widetilde{f}%
_{k}}{f_{k}}\cos \left( j\cdot \widetilde{\lambda }_{k}\right) +\frac{1}{2%
\mathbf{M}}\sum_{k\preceq M}\left( \frac{\widehat{f}_{k}-\widetilde{f}_{k}}{%
f_{k}}\right) \left( \frac{f_{k}-\widetilde{f}_{k}}{\widetilde{f}_{k}}%
\right) \cos \left( j\cdot \widetilde{\lambda }_{k}\right)  \notag \\
&&+o_{p}\left( \mathbf{n}^{-1/2}\right) \text{,}  \label{AC}
\end{eqnarray}%
after using Taylor series expansion of $\log \left( z\right) $ around $z=1$
and Condition $C4$. Now, the absolute value of the second term on the right
of the last displayed expression is bounded by 
\begin{equation}
\frac{1}{2\mathbf{M}}\sum_{k\preceq M}\left\vert \frac{\widehat{f}_{k}-%
\widetilde{f}_{k}}{f_{k}}\right\vert \left\vert \frac{f_{k}-\widetilde{f}_{k}%
}{\widetilde{f}_{k}}\right\vert =O\left( \frac{1}{\mathbf{Mm}^{1/2}}\right)
=o\left( \mathbf{n}^{-1/2}\right)  \label{AD}
\end{equation}%
by Lemmas \ref{c_hat} and \ref{f_hat}. So, we conclude that%
\begin{eqnarray*}
\mathbf{n}^{1/2}\left( \widehat{\alpha }_{j}-\widetilde{\alpha }%
_{j,n}\right) &=&\frac{\mathbf{n}^{1/2}}{2\mathbf{M}}\sum_{k\preceq M}\frac{%
\widehat{f}_{k}-\widetilde{f}_{k}}{f_{k}}\cos \left( j\cdot \widetilde{%
\lambda }_{k}\right) +o_{p}\left( 1\right) \\
&=&\frac{1}{2\mathbf{n}^{1/2}}\sum_{k\preceq n}\frac{I_{x}^{T}\left( \lambda
_{k}\right) -f\left( \lambda _{k}\right) }{f\left( \lambda _{k}\right) }%
h_{k,n}\left( j\right) +o_{p}\left( 1\right) \text{,}
\end{eqnarray*}%
where $h_{k,n}\left( j\right) $ is a step function defined as%
\begin{equation*}
h_{k,n}\left( j\right) =f_{p}^{-1}f\left( \lambda _{k}\right) \cos \left(
j\cdot \widetilde{\lambda }_{k}\right)
\end{equation*}%
when $2p\left[ \ell \right] -1<\frac{k\left[ \ell \right] }{m\left[ \ell %
\right] }<2p\left[ \ell \right] +1$ and $1\leq p\left[ 1\right] <M\left[ 1%
\right] $, $1-M\left[ 2\right] <p\left[ 2\right] \leq M\left[ 2\right] $.
Now, using Lemma 1, we have that for all $j$, 
\begin{equation*}
\sum_{k\preceq n}\left( \frac{I_{x}^{T}\left( \lambda _{k}\right) }{f\left(
\lambda _{k}\right) }-\frac{\left( 2\pi \right) ^{2}I_{\vartheta }^{T}\left(
\lambda _{k}\right) }{\sigma _{\vartheta }^{2}}\right) h_{k,n}\left(
j\right) =o_{p}\left( \mathbf{n}^{1/2}\right) \text{.}
\end{equation*}%
So, we conclude that the left side of $\left( \ref{pTh1.1}\right) $ is 
\begin{eqnarray*}
\mathbf{n}^{1/2}\sum_{j=1}^{J}\varphi _{j}\left( \widehat{\alpha }_{j}-%
\widetilde{\alpha }_{j,n}\right) &=&\sum_{j=1}^{J}\varphi _{j}\frac{1}{%
\mathbf{n}^{1/2}}\sum_{k\preceq n}\left( \frac{\left( 2\pi \right)
^{2}I_{\vartheta }^{T}\left( \lambda _{k}\right) }{\sigma _{\vartheta }^{2}}%
-1\right) h_{k,n}\left( j\right) +o_{p}\left( 1\right) \\
&=&\sum_{j=1}^{J}\varphi _{j}\frac{1}{\mathbf{n}^{1/2}}\sum_{k\preceq
n}\left( \frac{\left( 2\pi \right) ^{2}I_{\vartheta }^{T}\left( \lambda
_{k}\right) }{\sigma _{\vartheta }^{2}}-1\right) \cos \left( j\cdot 
\widetilde{\lambda }_{k}\right) \left( 1+o_{p}\left( 1\right) \right)
\end{eqnarray*}%
after we observe that Condition $C1$ implies that 
\begin{equation*}
h_{k,n}\left( j\right) =\cos \left( j\cdot \widetilde{\lambda }_{k}\right)
\left( 1+\frac{1}{\mathbf{M}^{1/2}}\left( \frac{\partial f\left( \lambda
_{k}\right) }{\partial \lambda \left[ 1\right] }+\frac{\partial f\left(
\lambda _{k}\right) }{\partial \lambda \left[ 2\right] }\right) +O\left( 
\frac{1}{\mathbf{M}}\right) \right) \text{.}
\end{equation*}%
Recall that $M\left[ 1\right] =M\left[ 2\right] $. From here the conclusion
is standard proceeding as in the proof of Theorems 1 and 2 of \cite%
{hidalgo2009goodness}, see also Robinson and Vidal-Sanz $\left( 2006\right) $%
, and so it is omitted.\hfill $\blacksquare $

\subsection{\textbf{Proof of Theorem 2}}

Define $\widehat{d}_{j}=\log \widehat{A}_{j}$, $\widetilde{d}_{j,n}=\log 
\widetilde{A}_{j,n}$ and $d_{j,n}=\log A_{j,n}$. We begin with part $\left( 
\mathbf{b}\right) $. First by definition, 
\begin{equation}
\widetilde{d}_{j,n}-d_{j,n}=:\sum_{k\preceq M}\left( \widetilde{\alpha }%
_{k,n}-\alpha _{k,n}\right) e^{-ik\cdot \widetilde{\lambda }_{j}}\text{,}
\label{pTh2.1}
\end{equation}%
which by Taylor expansion of $\log \left( \widetilde{f}_{r}/f_{r}\right) $, $%
\left( \ref{AA}\right) $ in Lemma \ref{c_hat} and Condition $C4$, it is%
\begin{eqnarray*}
&&\frac{1}{2\mathbf{M}}\sum_{k\preceq M}\sum_{r\preceq M}\left\{ \left( 
\frac{\widetilde{f}_{r}-f_{r}}{f_{r}}\right) +\frac{1}{2}\left( \frac{%
\widetilde{f}_{r}-f_{r}}{f_{r}}\right) ^{2}\right\} \cos \left( k\cdot 
\widetilde{\lambda }_{r}\right) e^{-ik\cdot \widetilde{\lambda }%
_{j}}+o\left( \frac{1}{\mathbf{m}^{1/2}}\right) \\
&=&\frac{1}{2\mathbf{M}}\sum_{k\preceq M}\left\{ \sum_{r\preceq M}\left\{
\left( \frac{\widetilde{f}_{r}-f_{r}}{f_{r}}\right) +\frac{1}{72\mathbf{M}%
^{2}}g_{r}^{2}\right\} \cos \left( k\cdot \widetilde{\lambda }_{r}\right)
e^{-ik\cdot \widetilde{\lambda }_{j}}\right\} +o\left( \frac{1}{\mathbf{m}%
^{1/2}}\right) \text{.}
\end{eqnarray*}%
Now using the inequality%
\begin{equation}
\left\vert \sum_{k\preceq M}e^{-ik\cdot \widetilde{\lambda }_{p}}\right\vert
\leq K\widetilde{\lambda }_{p\left[ 1\right] }^{-1}\widetilde{\lambda }_{p%
\left[ 2\right] }^{-1}\text{,}  \label{AB}
\end{equation}%
we have that 
\begin{eqnarray*}
\left\vert \frac{1}{\mathbf{M}^{3}}\sum_{k\preceq M}\left\{ \sum_{r\preceq
M}g_{r}^{2}\cos \left( k\cdot \widetilde{\lambda }_{r}\right) e^{-ik\cdot 
\widetilde{\lambda }_{j}}\right\} \right\vert &\leq &\frac{K}{\mathbf{M}^{3}}%
\sum_{r\preceq M}g_{r}^{2}\left\vert \sum_{k\preceq M}\cos \left( k\cdot 
\widetilde{\lambda }_{r}\right) e^{-ik\cdot \widetilde{\lambda }%
_{j}}\right\vert \\
&=&O\left( \frac{\log ^{2}\mathbf{M}}{\mathbf{M}^{2}}\right) \text{.}
\end{eqnarray*}

Thus using that $2\cos x=e^{ix}+e^{-ix}$, the right side of $\left( \ref%
{pTh2.1}\right) $ is 
\begin{eqnarray*}
&&\frac{1}{2\mathbf{M}}\sum_{k\preceq M}\left\{ \sum_{r\preceq M}\left( 
\frac{\widetilde{f}_{r}-f_{r}}{f_{r}}\right) \cos \left( k\cdot \widetilde{%
\lambda }_{r}\right) e^{-ik\cdot \widetilde{\lambda }_{j}}\right\} +O\left( 
\frac{\log ^{2}\mathbf{M}}{\mathbf{M}^{2}}\right) \\
&=&\frac{1}{2\mathbf{M}}\sum_{r\preceq M}\left\{ \left( \frac{\widetilde{f}%
_{r}-f_{r}}{f_{r}}\right) \sum_{k\preceq M}\left( \frac{e^{ik\cdot 
\widetilde{\lambda }_{r-j}}+e^{ik\cdot \widetilde{\lambda }_{-r-j}}}{2}%
\right) \right\} +O\left( \frac{\log ^{2}\mathbf{M}}{\mathbf{M}^{2}}\right)
\\
&=&\frac{1}{12\mathbf{M}}\left\{ \sum_{k\preceq M}\frac{1}{\mathbf{M}}%
\sum_{r\preceq M}g_{r}\cos \left( k\cdot \widetilde{\lambda }_{r}\right)
\right\} e^{-ik\cdot \widetilde{\lambda }_{j}}+O\left( \frac{1}{\mathbf{M}%
^{2}}\sum_{r\preceq M}\frac{1}{r\pm j}\right) \\
&=&\frac{1}{6\mathbf{M}}g_{j}+o\left( \frac{1}{\mathbf{m}^{1/2}}\right) 
\text{,}
\end{eqnarray*}%
where in the second equality we have used $\left( \ref{AA}\right) $ and $%
\left( \ref{AB}\right) $ and then Condition $C4$ and for\ the third equality
that 
\begin{equation*}
\frac{1}{\mathbf{M}}\sum_{r\preceq M}g_{r}\cos \left( k\cdot \widetilde{%
\lambda }_{r}\right) =\int g\left( \lambda \right) \cos \left( k\cdot
\lambda \right) d\lambda +O\left( \mathbf{M}^{-1}\right) =:\xi _{k}+O\left( 
\mathbf{M}^{-1}\right)
\end{equation*}%
and then that $\sum_{k\preceq M}\xi _{k}e^{-ik\cdot \lambda }=g\left(
\lambda \right) +O\left( \mathbf{M}^{-1}\right) $ since $g\left( \lambda
\right) $, given in $\left( \ref{g_1}\right) $, is twice continuously
differentiable so that $\left\vert \xi _{k}\right\vert =O\left( \left\vert
k\right\vert ^{-3}\right) $. From here we conclude the proof of part $\left( 
\mathbf{b}\right) $ by standard algebra.

Next, we show part $\left( \mathbf{a}\right) $. By Cram\'{e}r-Wold device,
it suffices to examine that for any set of finite constants $\varphi
_{q_{1}},...,\varphi _{q_{2}}$, the behaviour of 
\begin{equation*}
\mathbf{m}^{1/2}\sum_{j=q_{1}}^{q_{2}}\varphi _{j}\left( \widehat{A}%
_{j}-A_{j,n}\right) \text{.}
\end{equation*}%
First, by definitions of $\widehat{A}_{j,n}$ and $\widetilde{A}_{j,n}$, we
have that 
\begin{eqnarray}
\mathbf{m}^{1/2}\left( \widehat{d}_{j}-\widetilde{d}_{j,n}\right) &=&-%
\mathbf{m}^{1/2}\sum_{k\preceq M}\left( \widehat{\alpha }_{k}-\widetilde{%
\alpha }_{k,n}\right) e^{-ik\cdot \widetilde{\lambda }_{j}}  \label{Ad} \\
&=&-\mathbf{m}^{1/2}\sum_{k\preceq M}\frac{1}{2\mathbf{M}}\sum_{s\preceq
M}\left( \frac{\widehat{f}_{s}-\widetilde{f}_{s}}{\widetilde{f}_{s}}\right)
\cos \left( k\cdot \widetilde{\lambda }_{s}\right) e^{-ik\cdot \widetilde{%
\lambda }_{j}}+o_{p}\left( 1\right)  \notag \\
&=&\sum_{k\preceq M}\frac{\mathbf{m}^{1/2}}{2\mathbf{n}}\sum_{s\preceq
n}\rho _{s}h_{s,n}\left( k\right) e^{-ik\cdot \widetilde{\lambda }%
_{j}}+o_{p}\left( 1\right)  \notag
\end{eqnarray}%
proceeding as in the proof of Theorem 1, where $\rho _{s}=\left( 2\pi
\right) ^{2}\sigma _{\vartheta }^{-2}I_{\vartheta }^{T}\left( \lambda
_{s}\right) $ and $h_{s,n}\left( k\right) $ were defined there. So, because $%
n\left[ \ell \right] =2M\left[ \ell \right] m\left[ \ell \right] $ for $\ell
=1,2$, denoting $\psi _{s,n}\left( j\right) =\left( 2\mathbf{M}\right)
^{-1/2}\sum_{k\preceq M}h_{s,n}\left( k\right) e^{-ik\cdot \widetilde{%
\lambda }_{j}}$, we conclude that 
\begin{eqnarray}
\mathbf{m}^{1/2}\sum_{j=1}^{J}\varphi _{j}\left( \widehat{d}_{j}-\widetilde{d%
}_{j,n}\right) &=&\frac{1}{2\mathbf{n}^{1/2}}\sum_{s\preceq n}\rho
_{s}\sum_{j=1}^{J}\varphi _{j}\psi _{s,n}\left( j\right) +o_{p}\left(
1\right)  \label{pTh2.2} \\
&&\overset{d}{\rightarrow }\mathcal{N}\left( 0,\mathbf{\varphi }^{\prime }V%
\mathbf{\varphi }\right)  \notag
\end{eqnarray}%
proceeding as in the proof of Theorem 1, where $\mathbf{\varphi }^{\prime }%
\mathbf{=}\left( \varphi _{1},...,\varphi _{J}\right) $ and 
\begin{eqnarray*}
V &=&\lim_{n\rightarrow \infty }\sum_{j_{1},j_{2}=1}^{J}\varphi
_{j_{1}}\varphi _{j_{2}}\frac{1}{\mathbf{n}}\sum_{\ell \preceq n}\psi _{\ell
,n}\left( j_{1}\right) \overline{\psi _{\ell ,n}}\left( j_{2}\right) \\
&=&\lim_{n\rightarrow \infty }\sum_{j_{1},j_{2}=1}^{J}\varphi
_{j_{1}}\varphi _{j_{2}}\frac{1}{2\mathbf{M}^{2}}\sum_{\ell \preceq
M}\sum_{k_{1},k_{2}\preceq M}e^{-ik_{1}\widetilde{\lambda }_{j_{1}}+ik_{2}%
\widetilde{\lambda }_{j_{2}}}\cos \left( k_{1}\cdot \widetilde{\lambda }%
_{\ell }\right) \cos \left( k_{2}\cdot \widetilde{\lambda }_{\ell }\right) \\
&=&2^{-1}\sum_{j_{1},j_{2}=1}^{J}\varphi _{j_{1}}\varphi _{j_{2}}\left(
\delta _{j_{1}\left[ 1\right] -j_{2}\left[ 2\right] }+2^{-1}\phi _{j_{1}%
\left[ 1\right] }\phi _{j_{2}\left[ 1\right] }-i\phi _{j_{1}\left[ 1\right]
-j_{2}\left[ 1\right] }\right)
\end{eqnarray*}%
by Lemma \ref{var_1}. From here the conclusion of the theorem follows by
standard delta arguments.\hfill $\blacksquare $

$\left. {}\right. $

\subsection{\textbf{Proof of Theorem 3}}

We begin with part $\left( \mathbf{a}\right) $. To that end, it suffices to
show that 
\begin{equation}
\mathbf{n}^{1/2}\sum_{\upsilon =p}^{q}\varphi _{\upsilon }\left( \widehat{a}%
_{\upsilon }-\widetilde{a}_{\upsilon ,n}\right) \overset{d}{\rightarrow }%
\mathcal{N}\left( 0,\sum_{\upsilon _{1},\upsilon _{2}=p}^{q}\varphi
_{\upsilon _{1}}\varphi _{\upsilon _{2}}\Omega _{a,\upsilon _{1},\upsilon
_{2}}\right) \text{.}  \label{pTh3.1}
\end{equation}%
By definition of $\widehat{a}_{\upsilon }-\widetilde{a}_{\upsilon ,n}$ and
Taylor expansion of $\widehat{A}_{j}-\widetilde{A}_{j,n}$, a typical
component on the left of $\left( \ref{pTh3.1}\right) $ is 
\begin{equation}
\frac{\mathbf{n}^{1/2}}{4\mathbf{M}}\sum_{-M<j\leq M}\left( \widehat{d}_{j}-%
\widetilde{d}_{j,n}\right) \widetilde{A}_{j,n}e^{i\upsilon \cdot \widetilde{%
\lambda }_{j}}+\frac{\mathbf{n}^{1/2}}{4\mathbf{M}}\sum_{-M<j\leq
M}\left\vert \widehat{d}_{j}-\widetilde{d}_{j,n}\right\vert ^{2}\left\vert 
\widetilde{A}_{j,n}\right\vert \left( 1+o_{p}\left( 1\right) \right) \text{.}
\label{pTh3.2}
\end{equation}%
Now $\left( \ref{pTh2.2}\right) $ implies that 
\begin{equation*}
\mathbf{m}\left\vert \widehat{d}_{j}-\widetilde{d}_{j,n}\right\vert ^{2}=C%
\frac{1}{\mathbf{n}}\left\vert \sum_{k\preceq n}\rho _{k}\psi _{k,n}\left(
j\right) \right\vert ^{2}+o_{p}\left( 1\right) =O_{p}\left( 1\right) \text{,}
\end{equation*}%
by Theorem $A$ of \cite{Serfling1980}, p. 14 because Condition $C1$ implies
that $x^{4}\left( t\right) $ is uniformly integrable and Theorem 2 and the
continuous mapping theorem implies that $\left\vert \mathbf{n}%
^{-1/2}\sum_{k\preceq n}\rho _{k}\psi _{k,n}\left( j\right) \right\vert
^{2}\rightarrow _{d}\chi ^{2}$. So, using $\left( \ref{pTh2.1}\right) $ we
conclude that $\left( \ref{pTh3.2}\right) $ is%
\begin{eqnarray}
&&-\frac{\mathbf{n}^{1/2}}{4\mathbf{M}}\sum_{-M<j\leq M}\left(
\sum_{k\preceq M}\left( \widehat{\alpha }_{k}-\widetilde{\alpha }%
_{k,n}\right) e^{-ik\cdot \widetilde{\lambda }_{j}}\right) \widetilde{A}%
_{j,n}e^{i\upsilon \cdot \widetilde{\lambda }_{j}}+O_{p}\left( \frac{\mathbf{%
n}^{1/2}}{\mathbf{m}}\right)  \notag \\
&=&-\frac{\mathbf{n}^{1/2}}{4\mathbf{M}^{2}}\sum_{-M<j\leq M}\left(
\sum_{k,r\preceq M}\left( \frac{\widehat{f}_{r}-\widetilde{f}_{r}}{%
\widetilde{f}_{r}}\right) \cos \left( k\cdot \widetilde{\lambda }_{r}\right)
e^{-ik\cdot \widetilde{\lambda }_{j}}\right) \widetilde{A}_{j,n}e^{i\upsilon
\cdot \widetilde{\lambda }_{j}}+O_{p}\left( \frac{\mathbf{n}^{1/2}}{\mathbf{m%
}}\right) \text{.}  \notag \\
&=&-\frac{\mathbf{n}^{1/2}}{4\mathbf{M}}\sum_{r\preceq M}\left( \frac{%
\widehat{f}_{r}-\widetilde{f}_{r}}{\widetilde{f}_{r}}\right) \left(
\sum_{k\preceq M}\cos \left( k\cdot \widetilde{\lambda }_{r}\right)
a_{\upsilon -k}\right) +O_{p}\left( \mathbf{M}^{-1/2}\right)  \label{pTh3.3}
\\
&=&:\mathbf{n}^{1/2}\varrho _{n,\upsilon -k}+O_{p}\left( \mathbf{M}%
^{-1/2}\right) \text{,}  \notag
\end{eqnarray}%
where in the first equality we use $\left( \ref{Ad}\right) $ and in the
second equality that Lemma \ref{bias_2} implies that 
\begin{eqnarray*}
&&\frac{\mathbf{n}^{1/2}}{\mathbf{M}}\sum_{r\preceq M}\left( \frac{\widehat{f%
}_{r}-\widetilde{f}_{r}}{\widetilde{f}_{r}}\right) \sum_{k\preceq M}\cos
\left( k\cdot \widetilde{\lambda }_{r}\right) \left\{ \frac{1}{4\mathbf{M}}%
\sum_{-M<j\leq M}\widetilde{A}_{j,n}e^{-i\left( k-\upsilon \right) \cdot 
\widetilde{\lambda }_{j}}-a_{k-\upsilon }\right\} \\
&=&\frac{\mathbf{n}^{1/2}}{\mathbf{M}}\sum_{r\preceq M}\left( \frac{\widehat{%
f}_{r}-\widetilde{f}_{r}}{\widetilde{f}_{r}}\right) \sum_{k\preceq M}\cos
\left( k\cdot \widetilde{\lambda }_{r}\right) \left\{ \frac{1}{\mathbf{M}}%
h_{k-\upsilon }+O\left( \frac{\mathbf{1}}{\mathbf{M}^{2}}\right) \right\} \\
&=&O_{p}\left( \mathbf{M}^{-1/2}\right)
\end{eqnarray*}%
because $\left( \widehat{f}_{r}-\widetilde{f}_{r}\right) /\widetilde{f}%
_{r}=O_{p}\left( \mathbf{m}^{-1/2}\right) $ and $\left\{ h_{k}\right\} _{k}$
is a summable sequence.

So, we conclude that the left side of $\left( \ref{pTh3.1}\right) $ is 
\begin{equation*}
-\sum_{\upsilon =p}^{q}\varphi _{\upsilon }\mathbf{n}^{1/2}\varrho
_{n,\upsilon -k}+o_{p}\left( 1\right) \overset{d}{\rightarrow }\mathcal{N}%
\left( 0,V\right) \text{,}
\end{equation*}%
where $V=\lim_{M\rightarrow \infty }\mathbf{M}^{-1}\sum_{r\preceq M}\left(
\sum_{\upsilon =p}^{q}\varphi _{\upsilon }\sum_{k\preceq M}\cos \left(
k\cdot \lambda _{r}\right) a_{k-\upsilon }\right) ^{2}$.

We now conclude because using that $2\cos \left( x\right) =e^{ix}+e^{-ix}$, $%
2\cos \left( x\right) \cos \left( y\right) =\cos \left( x+y\right) +\cos
\left( x-y\right) $, a typical component of $V$ is $\varphi _{\upsilon
_{1}}\varphi _{\upsilon _{2}}\ $times%
\begin{eqnarray*}
&&\lim_{M\rightarrow \infty }\sum_{k_{1},k_{2}\preceq M}a_{k_{1}-\upsilon
_{1}}a_{k_{2}-\upsilon _{2}}\frac{1}{\mathbf{M}}\sum_{r\preceq M}\cos \left(
k_{1}\cdot \widetilde{\lambda }_{r}\right) \cos \left( k_{2}\cdot \widetilde{%
\lambda }_{r}\right) \\
&=&\lim_{M\rightarrow \infty }\sum_{k_{1},k_{2}\preceq M}a_{k_{1}-\upsilon
_{1}}a_{k_{2}-\upsilon _{2}}\frac{1}{2\mathbf{M}}\sum_{r\preceq M}\left(
\cos \left( \left( k_{1}+k_{2}\right) \cdot \widetilde{\lambda }_{r}\right)
+\cos \left( \left( k_{1}-k_{2}\right) \cdot \widetilde{\lambda }_{r}\right)
\right) \\
&=&\lim_{M\rightarrow \infty }\sum_{k_{1},k_{2}\preceq M}a_{k_{1}-\upsilon
_{1}}a_{k_{2}-\upsilon _{2}}\frac{1}{2M\left[ 1\right] }\sum_{r\left[ 1%
\right] =1}^{M\left[ 1\right] }\left\{ \cos \left( \left( k_{1}\left[ 1%
\right] +k_{2}\left[ 1\right] \right) \widetilde{\lambda }_{r\left[ 1\right]
}\right) \delta _{k_{2}\left[ 2\right] +k_{1}\left[ 2\right] }\right. \\
&&\text{ \ \ \ \ \ \ \ \ \ \ \ \ \ \ \ \ \ \ \ \ \ \ \ \ \ \ \ \ \ \ \ \ \ \
\ \ \ \ \ \ \ \ \ \ \ \ \ \ \ \ \ \ \ \ \ \ }\left. +\cos \left( \left( k_{1}%
\left[ 1\right] -k_{2}\left[ 1\right] \right) \widetilde{\lambda }_{r\left[ 1%
\right] }\right) \delta _{k_{2}\left[ 2\right] -k_{1}\left[ 2\right]
}\right\} \\
&=&\sum_{0\preceq k}a_{k}a_{k+\upsilon _{2}-\upsilon
_{1}}-\lim_{M\rightarrow \infty }\frac{1}{M\left[ 1\right] }\sum_{k_{1}\neq
k_{2}\preceq M;k_{1}\pm k_{2}\left[ 1\right] =1,3,...,\left[ M/2\right]
}a_{k_{1}-\upsilon _{1}}a_{k_{2}-\upsilon _{2}}\text{,}
\end{eqnarray*}%
where we have taken $\upsilon _{1}\preceq \upsilon _{2}$. From here the
conclusion is standard since $a_{\upsilon }$ is summable.

Part $\left( \mathbf{b}\right) $ follows by similar arguments to those in $%
\left( \ref{pTh3.3}\right) $ and Lemma \ref{c_hat}, so it is omitted.\hfill $%
\blacksquare $

$\left. {}\right. $

\subsection{\textbf{Proof of Theorem 4}}

We begin with part$\ \left( \mathbf{a}\right) $. For that purpose, denote 
\begin{eqnarray*}
\dot{x}_{s}^{\ast } &=&\sum_{k\left[ 2\right] =M\left[ 2\right] +1}^{\infty
}a_{0,k\left[ 2\right] }~x_{s\left[ 1\right] ,s\left[ 2\right] -k\left[ 2%
\right] }^{\ast }+\left( \sum_{k\left[ 1\right] =1}^{M\left[ 1\right]
}\sum_{k\left[ 2\right] =1-M\left[ 2\right] }^{M\left[ 2\right] }\right)
^{\perp }a_{k}x_{s-k}^{\ast } \\
\ddot{x}_{s}^{\ast } &=&\sum_{k\left[ 2\right] =1}^{M\left[ 2\right] }a_{0,k%
\left[ 2\right] }~x_{s\left[ 1\right] ,s\left[ 2\right] -k\left[ 2\right]
}^{\ast }+\sum_{k\left[ 1\right] =1}^{M\left[ 1\right] }\sum_{k\left[ 2%
\right] =1-M\left[ 2\right] }^{M\left[ 2\right] }a_{k}x_{s-k}^{\ast } \\
x_{n,s}^{\ast } &=&-\sum_{k\left[ 2\right] =1}^{M\left[ 2\right]
}a_{n,\left( 0,k\left[ 2\right] \right) }~x_{s\left[ 1\right] ,s\left[ 2%
\right] -k\left[ 2\right] }^{\ast }-\sum_{k\left[ 1\right] =1}^{M\left[ 1%
\right] }\sum_{k\left[ 2\right] =1-M\left[ 2\right] }^{M\left[ 2\right]
}a_{n,k}x_{s-k}^{\ast }\text{,}
\end{eqnarray*}%
where $\left( \sum_{k\left[ 1\right] =1}^{M\left[ 1\right] }\sum_{k\left[ 2%
\right] =1-M\left[ 2\right] }^{M\left[ 2\right] }\right) ^{\perp }=\sum_{k%
\left[ 1\right] =1}^{\infty }\sum_{k\left[ 2\right] =-\infty }^{\infty
}-\sum_{k\left[ 1\right] =1}^{M\left[ 1\right] }\sum_{k\left[ 2\right] =1-M%
\left[ 2\right] }^{M\left[ 2\right] }$. Then, 
\begin{equation}
x_{s}^{\ast }-\widehat{x}_{s}^{\ast }=\vartheta _{s}-\dot{x}_{s}^{\ast
}+\left( x_{n,s}^{\ast }-\widehat{x}_{s}^{\ast }\right) -\left(
x_{n,s}^{\ast }+\ddot{x}_{s}^{\ast }\right) \text{.}  \label{x_pred}
\end{equation}

The second moment of $\dot{x}_{s}^{\ast }$ is clearly $o\left( 1\right) $
since $\sum_{k\left[ 2\right] =-\infty }^{\infty }\left\vert a_{k\left[ 1%
\right] ,k\left[ 2\right] }\right\vert <K$ for any $k\left[ 1\right] $ and $%
\sum_{k\left[ 1\right] =1}^{\infty }\left\vert a_{k\left[ 1\right] ,k\left[ 2%
\right] }\right\vert <K$ for any $k\left[ 2\right] $ and that $M\rightarrow
\infty $. Next, the second moment of the last term on the right of $\left( %
\ref{x_pred}\right) $ is bounded by%
\begin{equation*}
2E\left( \sum_{k\left[ 2\right] =1}^{M\left[ 2\right] }\left( a_{n,\left( 0,k%
\left[ 2\right] \right) }-a_{0,k\left[ 2\right] }\right) ~x_{s\left[ 1\right]
,s\left[ 2\right] -k\left[ 2\right] }^{\ast }\right) ^{2}+2E\left( \sum_{k%
\left[ 1\right] =1}^{M\left[ 1\right] }\sum_{k\left[ 2\right] =1-M\left[ 2%
\right] }^{M\left[ 2\right] }\left( a_{n,k}-a_{k}\right) ~x_{s-k}^{\ast
}\right) ^{2}=o\left( 1\right) \text{,}
\end{equation*}%
by Lemma \ref{a_bias} and that the covariance of $x_{s}^{\ast }$ is
summable. Thus, it remains to examine the behaviour of $x_{n,s}^{\ast }-%
\widehat{x}_{s}^{\ast }$ on the right of $\left( \ref{x_pred}\right) $,
which is 
\begin{eqnarray}
&&\sum_{k\left[ 2\right] =1}^{M\left[ 2\right] }\left( \widehat{a}_{0,k\left[
2\right] }-\widetilde{a}_{n,\left( 0,k\left[ 2\right] \right) }\right) ~x_{s%
\left[ 1\right] ,s\left[ 2\right] -k\left[ 2\right] }^{\ast }+\sum_{k\left[ 1%
\right] =1}^{M\left[ 1\right] }\sum_{k\left[ 2\right] =1-M\left[ 2\right]
}^{M\left[ 2\right] }\left( \widehat{a}_{k}-\widetilde{a}_{n,k}\right)
x_{s-k}^{\ast }  \label{x_pred1} \\
&&+\sum_{k\left[ 2\right] =1}^{M\left[ 2\right] }\left( \widetilde{a}%
_{n,\left( 0,k\left[ 2\right] \right) }-a_{n,\left( 0,k\left[ 2\right]
\right) }\right) ~x_{s\left[ 1\right] ,s\left[ 2\right] -k\left[ 2\right]
}^{\ast }+\sum_{k\left[ 1\right] =1}^{M\left[ 1\right] }\sum_{k\left[ 2%
\right] =1-M\left[ 2\right] }^{M\left[ 2\right] }\left( \widetilde{a}%
_{n,k}-a_{n,k}\right) x_{s-k}^{\ast }\text{.}  \notag
\end{eqnarray}%
Now Theorem 3 part $\left( \mathbf{b}\right) $ and summability of the
covariance of $x_{s}^{\ast }$ yields that the second moment of the second
term of $\left( \ref{x_pred1}\right) $ is $o\left( 1\right) $. So, to
complete the proof of part $\left( \mathbf{a}\right) $, we need to look at
the first term, which is%
\begin{eqnarray*}
&&\sum_{k\left[ 2\right] =1}^{M\left[ 2\right] }\left( \widehat{a}_{0,k\left[
2\right] }-\widetilde{a}_{n,\left( 0,k\left[ 2\right] \right) }-\varrho
_{n,\left( 0,k\left[ 2\right] \right) }\right) ~x_{s\left[ 1\right] ,s\left[
2\right] -k\left[ 2\right] }^{\ast }+\sum_{k\left[ 1\right] =1}^{M\left[ 1%
\right] }\sum_{k\left[ 2\right] =1-M\left[ 2\right] }^{M\left[ 2\right]
}\left( \widehat{a}_{k}-\widetilde{a}_{n,k}-\varrho _{n,k}\right)
x_{s-k}^{\ast } \\
&&+\sum_{k\left[ 2\right] =1}^{M\left[ 2\right] }\varrho _{n,\left( 0,k\left[
2\right] \right) }~x_{s\left[ 1\right] ,s\left[ 2\right] -k\left[ 2\right]
}^{\ast }+\sum_{k\left[ 1\right] =1}^{M\left[ 1\right] }\sum_{k\left[ 2%
\right] =1-M\left[ 2\right] }^{M\left[ 2\right] }\varrho _{n,k}x_{s-k}^{\ast
} \\
&=&\sum_{k\left[ 2\right] =1}^{M\left[ 2\right] }\varrho _{n,\left( 0,k\left[
2\right] \right) }~x_{s\left[ 1\right] ,s\left[ 2\right] -k\left[ 2\right]
}^{\ast }+\sum_{k\left[ 1\right] =1}^{M\left[ 1\right] }\sum_{k\left[ 2%
\right] =1-M\left[ 2\right] }^{M\left[ 2\right] }\varrho _{n,k}x_{s-k}^{\ast
}+O_{p}\left( \mathbf{M}^{1/2}\mathbf{n}^{-1/2}\right)
\end{eqnarray*}%
because using expression $\left( \ref{pTh3.3}\right) $, $\left\vert \widehat{%
a}_{k}-a_{k,n}-\varrho _{k}\right\vert =O_{p}\left( \mathbf{n}^{-1/2}\mathbf{%
M}^{-1/2}\right) $. Finally%
\begin{equation*}
E\left\vert \sum_{k\left[ 1\right] =1}^{M\left[ 1\right] }\sum_{k\left[ 2%
\right] =1-M\left[ 2\right] }^{M\left[ 2\right] }\varrho _{n,k}x_{s-k}^{\ast
}\right\vert \leq \sum_{k\left[ 1\right] =1}^{M\left[ 1\right] }\sum_{k\left[
2\right] =1-M\left[ 2\right] }^{M\left[ 2\right] }\left( E\varrho
_{n,k}^{2}\right) ^{1/2}\left( Ex_{s-k}^{\ast 2}\right) ^{1/2}=O\left( 
\mathbf{Mn}^{-1/2}\right) =o\left( 1\right)
\end{equation*}%
by triangle and Cauchy-Schwarz inequalities and then Condition $C4$.
Similarly we have that $E\left\vert \sum_{k\left[ 2\right] =1}^{M\left[ 2%
\right] }\varrho _{n,0,k\left[ 2\right] }~x_{s\left[ 1\right] ,s\left[ 2%
\right] -k\left[ 2\right] }^{\ast }\right\vert =o\left( 1\right) $, which
concludes the proof of part $\left( \mathbf{a}\right) $.

We now show part $\left( \mathbf{b}\right) $, that is when $s=:\left( n\left[
1\right] +1,s\left[ 2\right] \right) $. For that purpose, it is convenient
to recall our representation in $\left( \ref{uni_1}\right) $. The reason is
because the prediction error $x_{n\left[ 1\right] +1,s\left[ 2\right]
}^{\ast }-\widehat{x}_{n\left[ 1\right] +1,s\left[ 2\right] }$ can be
written as 
\begin{equation*}
x_{n\left[ 1\right] +1,s\left[ 2\right] }^{\ast }-\widehat{x}_{n\left[ 1%
\right] +1,s\left[ 2\right] }=\vartheta _{s}^{\ast }+\sum_{0\prec k}\zeta
_{k}\vartheta _{s-k}^{\ast }-\sum_{k\left[ 2\right] =1}^{M\left[ 2\right] }%
\widehat{\zeta }_{0,k\left[ 2\right] }~\vartheta _{n\left[ 1\right] +1,s%
\left[ 2\right] -k\left[ 2\right] }^{\ast }-\sum_{k\left[ 1\right] =1}^{M%
\left[ 1\right] }\sum_{k\left[ 2\right] =1-M\left[ 2\right] }^{M\left[ 2%
\right] }\widehat{\zeta }_{k}~\vartheta _{s-k}^{\ast }\text{,}
\end{equation*}%
where $\widehat{\zeta }_{k\left[ 1\right] ,k\left[ 2\right] }$ is similar to 
$\widehat{a}_{k\left[ 1\right] ,k\left[ 2\right] }$ but where 
\begin{eqnarray*}
\widehat{\zeta }_{k} &=&\frac{1}{4\mathbf{M}}\sum_{-M<\ell \leq M}\widehat{A}%
_{\ell }e^{ik\cdot \widetilde{\lambda }_{\ell }}\text{, \ \ \ }\ k\in 
\mathcal{M}; \\
\widehat{A}_{\ell } &=&\overline{\widehat{A}}_{-\ell }=\exp \left\{ \left.
\sum_{j\preceq M}\right. ^{+}\widehat{\alpha }_{j}e^{-ij\cdot \widetilde{%
\lambda }_{\ell }}\right\} \text{, \ \ \ \ \ \ \ }\ell \in \mathcal{M\cup }%
\left\{ 0\right\}
\end{eqnarray*}%
and $\widehat{\alpha }_{j}$ as defined in $\left( \ref{cr_1}\right) $. Now,
it is obvious that we have the same type of (statistical) results for $%
\widehat{\zeta }_{k}$ as those obtained for $\widehat{a}_{k}$, and hence
proceeding as in part $\left( \mathbf{a}\right) $, we conclude that 
\begin{equation*}
x_{n\left[ 1\right] +1,s\left[ 2\right] }^{\ast }-\widehat{x}_{n\left[ 1%
\right] +1,s\left[ 2\right] }=\vartheta _{n\left[ 1\right] +1,s\left[ 2%
\right] }^{\ast }+\sum_{k\left[ 2\right] =1}^{\infty }\zeta _{0,k\left[ 2%
\right] }~\vartheta _{n\left[ 1\right] +1,s\left[ 2\right] -k\left[ 2\right]
}^{\ast }+o_{p}\left( 1\right)
\end{equation*}%
and that 
\begin{equation*}
AE\left( x_{n\left[ 1\right] +1,s\left[ 2\right] }^{\ast }-\widehat{x}_{n%
\left[ 1\right] +1,s\left[ 2\right] }\right) =\left( 1+\sum_{k\left[ 2\right]
=1}^{\infty }\zeta _{0,k\left[ 2\right] }^{2}\right) \sigma _{\vartheta }^{2}%
\text{.}
\end{equation*}%
This concludes the proof of the theorem.

\section{Technical Lemmas}

To simplify the notation, we abbreviate $\sum_{-m<j\leq m}$ by $\sum_{j}$\
in what follows.

\begin{lemma}
Under Conditions $C1-C4$ we have that 
\begin{equation*}
\widehat{f}_{k}=\frac{1}{4\mathbf{m}}\sum_{j}f\left( \lambda _{j}+\widetilde{%
\lambda }_{k}\right) \frac{I_{\vartheta }^{T}\left( \lambda _{j}+\widetilde{%
\lambda }_{k}\right) }{\sigma _{\vartheta }^{2}}+\epsilon _{n,k}\text{,}
\end{equation*}%
where $\left\{ \epsilon _{k,n}\right\} _{k}$ is a triangular array sequence
of r.v.'s such that $E\sup_{k}\left\vert \epsilon _{k,n}\right\vert
^{2}=o\left( \mathbf{m}^{-1}\right) $.
\end{lemma}

\begin{proof}
The proof follows easily from Lemma 4 of \cite{hidalgo2009goodness}, and so
it is omitted.
\end{proof}

\begin{lemma}
\label{c_hat}Assuming $C1-C4$, $\left( \mathbf{a}\right) $ $\widetilde{%
\alpha }_{k,n}-\alpha _{k,n}=\mathbf{M}^{-1}\mathbf{\xi }_{k}+O\left( 
\mathbf{M}^{-2}\right) $ and $\left( \mathbf{b}\right) $ $\alpha
_{k,n}-\alpha _{k}=O\left( \mathbf{M}^{-1}\right) $.
\end{lemma}

\begin{proof}
We begin with part $\left( \mathbf{a}\right) $. By definition of $\widetilde{%
f}_{j}$ and then Taylor series expansion of $\log \left( \cdot \right) $, we
have that%
\begin{equation*}
\widetilde{\alpha }_{k,n}-\alpha _{k,n}=\frac{1}{2\mathbf{M}}\sum_{j\preceq
M}\left\{ \left( \frac{\widetilde{f}_{j}-f_{j}}{f_{j}}\right) +\frac{1}{2}%
\left( \frac{\widetilde{f}_{j}-f_{j}}{f_{j}}\right) ^{2}\left( 1+o\left(
1\right) \right) \right\} \cos \left( k\cdot \widetilde{\lambda }_{j}\right) 
\text{.}
\end{equation*}%
So, it suffices to examine the behaviour of $f_{j}^{-1}\left( \widetilde{f}%
_{j}-f_{j}\right) $. By definition and $\left( \ref{g_1}\right) $, 
\begin{eqnarray}
\frac{\widetilde{f}_{j}-f_{j}}{f_{j}} &=&\frac{f_{j}^{-1}}{4\mathbf{m}}%
\sum_{k}\left\{ f\left( \lambda _{k+mj}\right) -f\left( \lambda _{mj}\right)
\right\}  \notag \\
&=&\frac{f_{j}^{-1}}{4\mathbf{m}}\sum_{k}\left\{ \frac{k^{2}\left[ 1\right] 
}{n^{2}\left[ 1\right] }f_{11}\left( \lambda _{mj}\right) +\frac{k^{2}\left[
2\right] }{n^{2}\left[ 2\right] }f_{22}\left( \lambda _{mj}\right) \right\}
+O\left( \frac{1}{\mathbf{M}^{2}}\right)  \notag \\
&=&\frac{1}{6\mathbf{M}}g_{j}+O\left( \frac{1}{\mathbf{M}^{2}}\right) \text{,%
}  \label{AA}
\end{eqnarray}%
because $f\left( \lambda \right) $ is a four times differentiable function
and $\sum_{k}k^{c_{1}}\left[ 1\right] k^{c_{2}}\left[ 2\right] =0$ if $%
c_{1}+c_{2}$ is an odd integer. From here the conclusion follows by standard
arguments, because $g\left( \lambda \right) $ is a continuous differentiable
function, so that the Riemman sums converge to their integral counterpart.

Part $\left( \mathbf{b}\right) $ follows using Lemma \ref{bias_2}.
\end{proof}

\begin{lemma}
\label{f_hat}Assuming, $C1-C4$, for all $k=1,2,...$%
\begin{equation*}
E\left( \widetilde{f}_{k}^{-1}\left( \widehat{f}_{k}-\widetilde{f}%
_{k}\right) \right) ^{2}=O\left( \mathbf{m}^{-1}\right) \text{.}
\end{equation*}
\end{lemma}

\begin{proof}
Because $\widetilde{f}_{k}=\left( 4\mathbf{m}\right) ^{-1}\sum_{j}f\left(
\lambda _{j+mk}\right) >0$, the left side of the last displayed equality is,
up to multiplicative constants, bounded by 
\begin{equation*}
E\left( \frac{1}{\mathbf{m}}\sum_{j}f\left( \lambda _{j+mk}\right) \left( 
\frac{I_{x}^{T}\left( \lambda _{j+mk}\right) }{f\left( \lambda
_{j+mk}\right) }-\frac{I_{\vartheta }^{T}\left( \lambda _{j+mk}\right) }{%
\sigma _{\vartheta }^{2}}\right) \right) ^{2}+E\left( \frac{1}{\mathbf{m}}%
\sum_{j}f\left( \lambda _{j+mk}\right) \left( \frac{I_{\vartheta }^{T}\left(
\lambda _{j+mk}\right) }{\sigma _{\vartheta }^{2}}-1\right) \right) ^{2}%
\text{.}
\end{equation*}%
The first term of the last displayed expression is $o\left( \mathbf{m}%
^{-1}\right) $ by Lemma 1, whereas the second term follows by standard
arguments, as $\vartheta _{t}$ is an $iid$ sequence of r.v.'s with finite
fourth moments.
\end{proof}

\begin{lemma}
\label{var_1} 
\begin{eqnarray}
&&\frac{1}{\mathbf{M}^{2}}\sum_{p\preceq M}\left( \sum_{k_{1}\preceq M}\cos
\left( k_{1}\cdot \widetilde{\lambda }_{p}\right) e^{-ik_{1}\cdot \widetilde{%
\lambda }_{j_{1}}}\right) \left( \sum_{k_{2}\preceq M}\cos \left(
-k_{2}\cdot \widetilde{\lambda }_{p}\right) e^{ik_{2}\cdot \widetilde{%
\lambda }_{j_{2}}}\right)  \label{lemma_41} \\
&=&2\left( \delta _{j_{1}\left[ 1\right] -j_{2}\left[ 1\right] }+2^{-1}\phi
_{j_{1}\left[ 1\right] }\phi _{j_{2}\left[ 1\right] }-i\phi _{j_{1}\left[ 1%
\right] -j_{2}\left[ 1\right] }\right) \delta _{j_{1}\left[ 2\right] \pm
j_{2}\left[ 2\right] }+O\left( \mathbf{M}^{-1}\right) \text{.}  \notag
\end{eqnarray}
\end{lemma}

\begin{proof}
First, 
\begin{eqnarray*}
\sum_{k\preceq M}e^{-ik\cdot \widetilde{\lambda }_{p}} &=&:\sum_{k\left[ 2%
\right] =1}^{M}e^{-ik\left[ 2\right] \frac{\pi p\left[ 2\right] }{M}}+\sum_{k%
\left[ 1\right] =1}^{M}e^{-ik\left[ 1\right] \frac{\pi p\left[ 1\right] }{M}%
}\sum_{k\left[ 2\right] =1-M}^{M}e^{-ik\left[ 2\right] \frac{\pi p\left[ 2%
\right] }{M}} \\
&=&\sum_{k\left[ 2\right] =1}^{M}e^{-ik\left[ 2\right] \frac{\pi p\left[ 2%
\right] }{M}}+2M\sum_{k\left[ 1\right] =1}^{M}e^{-ik\left[ 1\right] \frac{%
\pi p\left[ 1\right] }{M}}\delta _{p\left[ 2\right] } \\
&=&:\mathcal{D}\left( p\left[ 2\right] \right) +2M\mathcal{D}\left( p\left[ 1%
\right] \right) \delta _{p\left[ 2\right] }=\Xi \left( p\left[ 1\right] ,p%
\left[ 2\right] \right) \text{,}
\end{eqnarray*}%
where for notational simplicity, we assume that $M\left[ 1\right] =M\left[ 2%
\right] =:M$.

Next, because $2\cos \left( x\right) =\exp \left( ix\right) +\exp \left(
-ix\right) $, we have then that $\left( \ref{lemma_41}\right) $ is 
\begin{eqnarray}
&&\frac{1}{4\mathbf{M}^{2}}\sum_{p\preceq M}\left( \sum_{k_{1}\preceq
M}\left( e^{-ik_{1}\cdot \widetilde{\lambda }_{p+j_{1}}}+e^{-ik_{1}\cdot 
\widetilde{\lambda }_{j_{1}-p}}\right) \right) \left( \sum_{k_{2}\preceq
M}\left( e^{ik_{2}\cdot \widetilde{\lambda }_{p+j_{2}}}+e^{-ik_{2}\cdot 
\widetilde{\lambda }_{p-j_{2}}}\right) \right)  \notag \\
&=&\frac{1}{4\mathbf{M}^{2}}\sum_{p\preceq M}\left\{ \left( \Xi \left( p%
\left[ 1\right] +j_{1}\left[ 1\right] ,p\left[ 2\right] +j_{1}\left[ 2\right]
\right) +\Xi \left( j_{1}\left[ 1\right] -p\left[ 1\right] ,j_{1}\left[ 2%
\right] -p\left[ 2\right] \right) \right) \right.  \notag \\
&&\text{ \ \ \ \ \ \ \ \ \ \ \ \ \ \ \ \ }\left. \left( \Xi \left( -p\left[ 1%
\right] -j_{2}\left[ 1\right] ,-p\left[ 2\right] -j_{2}\left[ 2\right]
\right) +\Xi \left( p\left[ 1\right] -j_{2}\left[ 1\right] ,p\left[ 2\right]
-j_{2}\left[ 2\right] \right) \right) \right\} \text{. }  \label{lemma_42}
\end{eqnarray}

Let's examine a typical term on the right of $\left( \ref{lemma_42}\right) $%
, say%
\begin{equation*}
\frac{1}{4\mathbf{M}^{2}}\sum_{p\preceq M}\Xi \left( p\left[ 1\right] +j_{1}%
\left[ 1\right] ,p\left[ 2\right] +j_{1}\left[ 2\right] \right) ~\Xi \left(
-p\left[ 1\right] -j_{2}\left[ 1\right] ,-p\left[ 2\right] -j_{2}\left[ 2%
\right] \right) \text{.}
\end{equation*}%
By definition, the last displayed expression is 
\begin{eqnarray*}
&&\frac{1}{4\mathbf{M}^{2}}\sum_{p\preceq M}\left\{ \mathcal{D}\left( p\left[
2\right] +j_{1}\left[ 2\right] \right) \mathcal{D}\left( -p\left[ 2\right]
-j_{2}\left[ 2\right] \right) \right\} \\
&&+\frac{1}{2M^{3}}\sum_{p\preceq M}\left\{ \mathcal{D}\left( p\left[ 2%
\right] +j_{1}\left[ 2\right] \right) \mathcal{D}\left( -p\left[ 1\right]
-j_{2}\left[ 1\right] \right) \delta _{p\left[ 2\right] +j_{2}\left[ 2\right]
}\right\} \\
&&+\frac{1}{2M^{3}}\sum_{p\preceq M}\left\{ \mathcal{D}\left( -p\left[ 2%
\right] -j_{2}\left[ 2\right] \right) \mathcal{D}\left( p\left[ 1\right]
+j_{1}\left[ 1\right] \right) \delta _{p\left[ 2\right] +j_{1}\left[ 2\right]
}\right\} \\
&&+\frac{1}{M^{2}}\sum_{p\preceq M}\left\{ \mathcal{D}\left( p\left[ 1\right]
+j_{1}\left[ 1\right] \right) \delta _{p\left[ 2\right] +j_{1}\left[ 2\right]
}\mathcal{D}\left( -p\left[ 1\right] -j_{2}\left[ 1\right] \right) \delta _{p%
\left[ 2\right] +j_{2}\left[ 2\right] }\right\}
\end{eqnarray*}

Because $\left( \ref{AB}\right) $, it easy to see that the first term is $%
O\left( M^{-1}\right) $, whereas the second term is 
\begin{eqnarray*}
&&\frac{1}{2M^{3}}\sum_{p\left[ 1\right] =1}^{M}\mathcal{D}\left( j_{1}\left[
2\right] -j_{2}\left[ 2\right] \right) \mathcal{D}\left( p\left[ 1\right]
+j_{2}\left[ 1\right] \right) \\
&\leq &K\frac{1}{M^{2}}\sum_{p\left[ 1\right] =1}^{M}\mathcal{D}\left( p%
\left[ 1\right] +j_{2}\left[ 1\right] \right) \leq K\frac{1}{M}\sum_{p\left[
1\right] =1}^{M}\frac{1}{\left( p\left[ 1\right] +j_{2}\left[ 1\right]
\right) _{+}} \\
&=&O\left( \frac{\log M}{M}\right) \text{,}
\end{eqnarray*}%
so is the third term by symmetry. Finally the fourth term is different than
zero if $j_{1}\left[ 2\right] =j_{2}\left[ 2\right] $, in which case becomes%
\begin{equation*}
\frac{1}{M^{2}}\sum_{p\left[ 1\right] =1}^{M}\mathcal{D}\left( p\left[ 1%
\right] +j_{1}\left[ 1\right] \right) \mathcal{D}\left( -p\left[ 1\right]
-j_{2}\left[ 1\right] \right) \text{.}
\end{equation*}%
Then, proceeding similarly with the other three terms in $\left( \ref%
{lemma_42}\right) $, we can conclude, except negligible terms, that it is%
\begin{eqnarray*}
&&\frac{1}{M^{2}}\sum_{p\left[ 1\right] =1}^{M}\left\{ \mathcal{D}\left( p%
\left[ 1\right] +j_{1}\left[ 1\right] \right) \mathcal{D}\left( -p\left[ 1%
\right] -j_{2}\left[ 1\right] \right) +\mathcal{D}\left( j_{1}\left[ 1\right]
-p\left[ 1\right] \right) \mathcal{D}\left( p\left[ 1\right] -j_{2}\left[ 1%
\right] \right) \right\} \delta _{j_{1}\left[ 2\right] -j_{2}\left[ 2\right]
} \\
&&\frac{1}{M^{2}}\sum_{p\left[ 1\right] =1}^{M}\left\{ \mathcal{D}\left(
j_{1}\left[ 1\right] -p\left[ 1\right] \right) \mathcal{D}\left( -p\left[ 1%
\right] -j_{2}\left[ 1\right] \right) +\mathcal{D}\left( p\left[ 1\right]
+j_{1}\left[ 1\right] \right) \mathcal{D}\left( p\left[ 1\right] -j_{2}\left[
1\right] \right) \right\} \delta _{j_{1}\left[ 2\right] +j_{2}\left[ 2\right]
} \\
&=&\frac{2}{M^{2}}\sum_{k_{1},k_{2}\preceq M}\left( e^{-i\left( j_{1}\left[ 1%
\right] \widetilde{\lambda }_{k_{1}\left[ 1\right] }-j_{2}\left[ 1\right] 
\widetilde{\lambda }_{k_{2}\left[ 1\right] }\right) }\sum_{p\left[ 1\right]
=1}^{M}\cos \left( \left( k_{1}\left[ 1\right] -k_{2}\left[ 1\right] \right) 
\widetilde{\lambda }_{p\left[ 1\right] }\right) \right) \delta _{j_{1}\left[
2\right] -j_{2}\left[ 2\right] } \\
&&+\frac{2}{M^{2}}\sum_{k_{1},k_{2}\preceq M}\left( e^{-i\left( j_{1}\left[ 1%
\right] \widetilde{\lambda }_{k_{1}\left[ 1\right] }-j_{2}\left[ 1\right] 
\widetilde{\lambda }_{k_{2}\left[ 1\right] }\right) }\sum_{p\left[ 1\right]
=1}^{M}\cos \left( \left( k_{1}\left[ 1\right] +k_{2}\left[ 1\right] \right) 
\widetilde{\lambda }_{p\left[ 1\right] }\right) \right) \delta _{j_{1}\left[
2\right] +j_{2}\left[ 2\right] }\text{.}
\end{eqnarray*}%
From here we conclude by Lemma 4 of \cite{Hidalgo2002}.
\end{proof}

\begin{lemma}
\label{bias_1}Under Condition $C1$, we have that%
\begin{equation*}
\sum_{\left\{ M\left[ 1\right] \leq k\left[ 1\right] \right\} \vee \left\{ M%
\left[ 2\right] \leq k\left[ 2\right] \right\} }\alpha _{k}e^{-ik\cdot 
\widetilde{\lambda }_{j}}=O\left( \mathbf{M}^{-2}\right) \text{.}
\end{equation*}
\end{lemma}

\begin{proof}
The proof is standard because four times continuous differentiability of $%
f\left( \lambda \right) $ implies that $\alpha _{k\left[ \ell \right]
}=O\left( k\left[ \ell \right] ^{-5}\right) $ for $\ell =1,2$.
\end{proof}

The next lemma is regarding the approximation of integrals by sums. Taking
for simplicity that $\check{n}=:n\left[ 1\right] =n\left[ 2\right] $ and
recalling our notation, we have then that $j/n=:\left( j\left[ 1\right] /%
\check{n},j\left[ 2\right] /\check{n}\right) $. Also, use the standard
notation, $\left\vert k\right\vert =k_{1}+k_{2}$, $k!=k_{1}!k_{1}!$, $%
y^{k}=y_{1}^{k_{1}}y_{2}^{k_{2}}$ and for a function $\Upsilon \left(
x\right) $ 
\begin{equation*}
\partial ^{k}\Upsilon \left( x\right) =\frac{\partial ^{\left\vert
k\right\vert }\Upsilon \left( x\right) }{\partial x^{\left\vert k\right\vert
}}\text{.}
\end{equation*}

\begin{lemma}
\label{bias_2}Assume that $\Upsilon \left( \cdot \right) $ is a function $q$
times continuously differentiable in $\left[ 0,1\right] ^{2}$. Then,%
\begin{equation}
\frac{1}{\check{n}^{2}}\sum_{j\left[ 1\right] =1}^{\check{n}}\sum_{j\left[ 2%
\right] =1}^{\check{n}}\Upsilon \left( \frac{j}{n}\right)
-\int_{0}^{1}\int_{0}^{1}\Upsilon \left( x\right) dx=\sum_{\left\vert
k\right\vert \leq q-1}h_{n\mathbf{,}k}\digamma _{k}+O\left( \frac{1}{\check{n%
}^{q}}\right) \text{,}  \label{E_M}
\end{equation}%
where $h_{n\mathbf{,}k}$ is a sequence such that $h_{n\mathbf{,}k}=O\left( 
\check{n}^{-k}\right) $ and $\digamma _{1}$, $\digamma _{2}$,...,$\digamma
_{q-1}$ are finite constants.
\end{lemma}

\begin{proof}
The left side of $\left( \ref{E_M}\right) $ is%
\begin{eqnarray}
&&\sum_{j\left[ 1\right] =1}^{\check{n}}\sum_{j\left[ 2\right] =1}^{\check{n}%
}\int_{\frac{j\left[ 1\right] -1}{\check{n}}}^{\frac{j\left[ 1\right] }{%
\check{n}}}\int_{\frac{j\left[ 2\right] -1}{\check{n}}}^{\frac{j\left[ 2%
\right] }{\check{n}}}\left( \Upsilon \left( \frac{j}{n}\right) -\Upsilon
\left( x\right) \right) dx  \label{E_M1} \\
&=&\sum_{j\left[ 1\right] =1}^{\check{n}}\sum_{j\left[ 2\right] =1}^{\check{n%
}}\int_{\frac{j\left[ 1\right] -1}{\check{n}}}^{\frac{j\left[ 1\right] }{%
\check{n}}}\int_{\frac{j\left[ 2\right] -1}{\check{n}}}^{\frac{j\left[ 2%
\right] }{\check{n}}}\left\{ \sum_{\left\vert k\right\vert \leq q-1}\frac{%
\partial ^{k}\Upsilon \left( \frac{j}{n}\right) }{k!}\left( x-\frac{j}{n}%
\right) ^{k}+\sum_{\left\vert k\right\vert =q}\frac{\partial ^{k}\Upsilon
\left( x\left( j\right) \right) }{k!}\left( x-\frac{j}{n}\right)
^{k}\right\} dx\text{,}  \notag
\end{eqnarray}%
by Taylor's expansion and where $x\left( j\right) $ denotes a point between $%
\left( j-1\right) /n$ and $j/n$. Now, the right side of $\left( \ref{E_M1}%
\right) $ is 
\begin{equation}
\sum_{\left\vert k\right\vert \leq q-1}\frac{1}{\check{n}^{k+2}}\sum_{j\left[
1\right] =1}^{\check{n}}\sum_{j\left[ 2\right] =1}^{\check{n}}\frac{\partial
^{k}\Upsilon \left( \frac{j}{n}\right) }{k!}+O\left( \frac{1}{\check{n}^{q}}%
\right) \text{.}  \label{E_M2}
\end{equation}

Denoting $\partial ^{k}\Upsilon \left( x\right) =\eta _{k}\left( x\right) $
and $\digamma _{k}=\int_{0}^{1}\int_{0}^{1}\eta _{k}\left( x\right) dx$, $%
k=1,...,3$, and proceeding as with $\left( \ref{E_M1}\right) $, we conclude
that $\left( \ref{E_M2}\right) $, and hence the left side of $\left( \ref%
{E_M}\right) $, is%
\begin{equation*}
\sum_{\left\vert k\right\vert \leq q-1}\frac{K}{\check{n}^{k}}\digamma
_{k}+O\left( \frac{1}{\check{n}^{4}}\right)
\end{equation*}%
after observing that $\eta _{k}\left( x\right) $, $k=1,...,q-1$, are
respectively $q-k$ continuous differentiable functions.
\end{proof}

\begin{lemma}
\label{a_bias}Under Condition $C1$, for all $p$, 
\begin{equation*}
a_{p,n}-a_{p}=\frac{\varrho _{p}}{\mathbf{M}}+O\left( \mathbf{M}^{-2}\log 
\mathbf{M}\right) \text{,}
\end{equation*}%
where $\left\{ \left\vert \varrho _{p}\right\vert \right\} _{p\geq 1}$ is a
summable sequence.
\end{lemma}

\begin{proof}
By definition, $a_{p,n}-a_{p}$ is%
\begin{eqnarray}
&&\frac{1}{4\mathbf{M}}\sum_{-M<j\leq M}\left( A_{j,n}-A_{j}^{\ast }\right)
e^{ip\cdot \widetilde{\lambda }_{j}}+\frac{1}{4\mathbf{M}}\sum_{-M<j\leq
M}\left( A_{j}^{\ast }-A_{j}\right) e^{ip\cdot \widetilde{\lambda }_{j}} 
\notag \\
&&+\left( \frac{1}{4\mathbf{M}}\sum_{-M<j\leq M}A_{j}e^{ip\cdot \widetilde{%
\lambda }_{j}}-\frac{1}{4\pi ^{2}}\int_{\Pi ^{2}}A\left( \lambda \right)
e^{ip\cdot \lambda }d\lambda \right) \text{,}  \label{pLem6.11}
\end{eqnarray}%
where $A_{j}^{\ast }=\exp \left( -\sum_{k\preceq M}\alpha _{k}e^{-ik\cdot 
\widetilde{\lambda }_{j}}\right) $. Because $A\left( \lambda \right) $ is
four times continuous differentiable, the third term of $\left( \ref%
{pLem6.11}\right) $ is 
\begin{equation}
\mathbf{M}^{-1}\digamma _{2,p}+\mathbf{M}^{-3/2}\digamma _{3,p}+O\left( 
\mathbf{M}^{-2}\right)  \label{BB1}
\end{equation}%
by Lemma \ref{bias_2}. This is the case after we notice that $\digamma
_{\ell ,p}$ there is given by%
\begin{equation*}
\digamma _{\ell ,p}=:\int_{-\pi }^{\pi }\int_{-\pi }^{\pi }\partial ^{\ell
}A\left( \lambda \right) e^{ip\cdot \lambda _{j}}d\lambda =O\left( p^{\ell
}a_{p}\right) =O\left( p^{\ell -5}\right)
\end{equation*}%
which is clearly summable since $\ell \leq 3$, and because $A\left( -\pi
,\lambda \left[ 2\right] \right) =A\left( \pi ,\lambda \left[ 2\right]
\right) $ and $A\left( \lambda \left[ 2\right] ,-\pi \right) =A\left(
\lambda \left[ 2\right] ,\pi \right) $ for every $\lambda \left[ 1\right]
,\lambda \left[ 2\right] \in \Pi $ implies that 
\begin{equation*}
\digamma _{1,p}=:\int_{-\pi }^{\pi }\int_{-\pi }^{\pi }\left( \frac{\partial
A\left( \lambda \right) e^{ip\cdot \lambda _{j}}}{\partial \lambda \left[ 1%
\right] }+\frac{\partial A\left( \lambda \right) e^{ip\cdot \lambda _{j}}}{%
\partial \lambda \left[ 2\right] }\right) d\lambda =0\text{.}
\end{equation*}

Next, the second term of $\left( \ref{pLem6.11}\right) $ is bounded in
absolute value by%
\begin{eqnarray}
&&\left\vert \frac{1}{4\mathbf{M}}\sum_{-M<j\leq M}\left( \exp \left\{
\left. \sum_{p}\right. ^{\dagger }\alpha _{k}e^{-ik\cdot \widetilde{\lambda }%
_{j}}\right\} -1\right) A_{j}e^{ip\cdot \widetilde{\lambda }_{j}}\right\vert
\notag \\
&\leq &\frac{K}{4\mathbf{M}}\sum_{-M<j\leq M}\left\vert A_{j}\right\vert
\left\vert \left. \sum_{k}\right. ^{\dagger }\alpha _{k}e^{-ik\cdot 
\widetilde{\lambda }_{j}}\right\vert \left( 1+O\left( 1\right) \right)
=O\left( \mathbf{M}^{-2}\right) \text{,}  \label{BB2}
\end{eqnarray}%
by Lemma \ref{bias_1} and that $\sum_{-M<j\leq M}\left\vert A_{j}\right\vert
=O\left( \mathbf{M}\right) $, where $\sum_{p}^{\dagger }$ denotes the
summation in 
\begin{equation*}
\mathcal{S}\left( p\right) =:\left\{ p:\left( 0\prec p\right) \wedge \left\{
\left( M\left[ 1\right] <p\left[ 1\right] \right) \vee \left( M\left[ 2%
\right] <p\left[ 2\right] \right) \right\} \right\} \text{.}
\end{equation*}

Finally, by definition of $A_{j,n}$ and $A_{j}^{\ast }$, the first term of $%
\left( \ref{pLem6.11}\right) $ is%
\begin{equation}
\frac{1}{2\mathbf{M}}\sum_{-M<j\leq M}\left( \exp \left\{ \sum_{k\preceq
M}\left( \alpha _{k,n}-\alpha _{k}\right) e^{-ik\cdot \widetilde{\lambda }%
_{j}}\right\} -1\right) A_{j}^{\ast }e^{ip\cdot \widetilde{\lambda }_{j}}%
\text{,}  \label{B.7}
\end{equation}%
where using the inequality in $\left( \ref{AB}\right) $, we have that $%
\sum_{k\preceq M}\left( \alpha _{k,n}-\alpha _{k}\right) e^{-ik\cdot 
\widetilde{\lambda }_{j}}$ is 
\begin{eqnarray}
&&\sum_{-k\preceq M}\left( \frac{1}{4\mathbf{M}}\sum_{-M<\ell \leq M}\log
\left( f_{\ell }\right) \cos \left( k\cdot \widetilde{\lambda }_{\ell
}\right) -\frac{1}{4\pi ^{2}}\int_{\Pi ^{2}}\log \left( f\left( \lambda
\right) \right) \cos \left( k\cdot \lambda \right) d\lambda \right)
e^{-ik\cdot \widetilde{\lambda }_{j}}  \notag \\
&=&\sum_{k\preceq M}\left( \frac{1}{\mathbf{M}}\digamma _{2,k}+\frac{1}{%
\mathbf{M}^{3/2}}\digamma _{3,k}\right) e^{-ik\cdot \widetilde{\lambda }%
_{j}}+O\left( \mathbf{M}^{-1}j^{-1}\right) \text{,}  \label{B.8}
\end{eqnarray}%
by Lemma \ref{bias_2} and where now $\digamma _{\ell ,k}=\int_{\Pi ^{2}}\log
\left( f\left( \lambda \right) \right) \cos \left( k\cdot \lambda \right)
d\lambda $, $\ell =2,3$. Next, because $\log \left( f\left( \lambda \right)
\right) $ is four times continuously differentiable, it implies that $%
\digamma _{\ell ,p}=:O\left( p^{\ell -5}\right) $ and thus by standard
arguments, the right side of $\left( \ref{B.8}\right) $ is 
\begin{eqnarray*}
&&\mathbf{M}^{-1}\psi _{2}\left( \widetilde{\lambda }_{j}\right) +\mathbf{M}%
^{-3/2}\psi _{3}\left( \widetilde{\lambda }_{j}\right) +O\left( \mathbf{M}%
^{-1}j^{-1}\right) \\
\psi _{\ell }\left( \lambda \right) &=&\sum_{k\preceq \infty }\digamma
_{\ell ,k}e^{-ik\cdot \lambda }\text{.}
\end{eqnarray*}%
From here and using Taylor expansion of $\exp (x)$, we obtain that $\left( %
\ref{B.7}\right) $\ is%
\begin{eqnarray}
&&\frac{1}{2\mathbf{M}}\sum_{-M<j\leq M}\left\{ \mathbf{M}^{-1}\psi
_{2}\left( \widetilde{\lambda }_{j}\right) +\mathbf{M}^{-3/2}\psi _{3}\left( 
\widetilde{\lambda }_{j}\right) \right\} A_{j}^{\ast }e^{ip\cdot \widetilde{%
\lambda }_{j}}+O\left( \mathbf{M}^{-2}\log \mathbf{M}\right)  \notag \\
&=&\frac{\nu _{p,2}}{\mathbf{M}}+\frac{\nu _{p,3}}{\mathbf{M}^{3/2}}+O\left( 
\mathbf{M}^{-2}\log \mathbf{M}\right)  \label{BB3}
\end{eqnarray}%
where $\left\{ \nu _{p,\ell }\right\} _{p\geq 1}$, $\ell =2,3$, are the
Fourier coefficients of $\psi _{\ell }\left( \lambda \right) A^{\ast }\left(
\lambda \right) $, which are summable because $\psi _{\ell }\left( \lambda
\right) A^{\ast }\left( \lambda \right) $ is $4-\ell $ times differentiable
function. The conclusion of\ the lemma now follows by gathering terms $%
\left( \ref{BB1}\right) $, $\left( \ref{BB2}\right) $ and $\left( \ref{BB3}%
\right) $.
\end{proof}

\newpage 
\bibliographystyle{chicago}
\bibliography{pred_references}

\end{document}